\documentclass[journal, 10pt]{IEEEtran}
\usepackage{cite}
\usepackage[pdftex]{graphicx}
\usepackage{epstopdf}
\graphicspath{{Figures/}}
\usepackage[cmex10]{amsmath}
\usepackage[tight,footnotesize]{subfigure}

\usepackage{stfloats}
\usepackage{url}
\usepackage{amsmath,amsfonts,amssymb,epsfig}
\usepackage{amsthm}
\newcounter{opteq}
\newenvironment{opteq}{\refstepcounter{opteq}\align}{\tag{P\theopteq}\endalign}
\usepackage{todonotes}
\let\labelindent\relax \usepackage{enumitem}

\usepackage[capitalize]{cleveref}

\newtheorem{definition}{Definition}
\newtheorem{lemma}{Lemma}
\newtheorem{proposition}{Proposition}
\newtheorem{thm}{Theorem}
\newtheorem{corollary}{Corollary}

\crefname{thm}{Theorem}{Theorems}

\usepackage{algorithm}
\usepackage{algorithmicx}
\usepackage[noend]{algpseudocode}

\usepackage{mathtools}
\usepackage{bbm}
\usepackage{arydshln}
\usepackage{mathrsfs}
\usepackage{bm}
\usepackage{pifont}

\usepackage{threeparttable}
\usepackage{multirow}
\usepackage[makeroom]{cancel}

\newcommand{\txt}[1]{\text{\normalfont #1}}

\DeclareMathOperator*{\argmax}{arg\,max}

\title{Sparse Symmetric Linear Arrays with Low Redundancy and a Contiguous Sum Co-Array}

\author{Robin~Rajam\"{a}ki,~\IEEEmembership{Student Member,~IEEE,}
	and~Visa~Koivunen,~\IEEEmembership{Fellow,~IEEE}
	\thanks{The authors are with the Department of Signal Processing and Acoustics, Aalto University, Espoo, Finland (e-mail: robin.rajamaki@aalto.fi, visa.koivunen@aalto.fi).}
	\thanks{This work was supported by the Academy of Finland project \emph{Massive and Sparse Antenna Array Processing for Millimeter-Wave Communications}.}}%

\begin{document}
\maketitle

\begin{abstract}
Sparse arrays can resolve significantly more scatterers or sources than sensor by utilizing the co-array --- a virtual array structure consisting of pairwise differences or sums of sensor positions. Although several sparse array configurations have been developed for passive sensing applications, far fewer active array designs exist. In active sensing, the sum co-array is typically more relevant than the difference co-array, especially when the scatterers are fully coherent. This paper proposes a general symmetric array configuration suitable for both active and passive sensing. We first derive necessary and sufficient conditions for the sum and difference co-array of this array to be contiguous. We then study two specific instances based on the Nested array and the Kl{\o}ve-Mossige basis, respectively. In particular, we establish the relationship between the minimum-redundancy solutions of the two resulting symmetric array configurations, and the previously proposed Concatenated Nested Array (CNA) and Kl{\o}ve Array (KA). Both the CNA and KA have closed-form expressions for the sensor positions, which means that they can be easily generated for any desired array size. The two array structures also achieve low redundancy, and a contiguous sum and difference co-array, which allows resolving vastly more scatterers or sources than sensors.
\end{abstract}

\section{Introduction}
\emph{Sensor arrays} are a key technology in for example, radar, wireless communication, medical imaging, radio astronomy, sonar, and seismology \cite{vantrees2002optimum}. The key advantages of arrays include spatial selectivity and the capability to mitigate interference. However, conventional \emph{uniform array} configurations may become prohibitively expensive, when a high spatial resolution facilitated by a large electrical aperture, and consequently a large number of sensors is required. 

\emph{Sparse arrays} allow for significantly reducing the number of sensors and costly RF-IF chains, whilst resolving vastly more scatterers or signal sources than sensors. This is facilitated by a virtual array model called the \emph{co-array} \cite{haubrich1968array,hoctor1990theunifying}, which is commonly defined in terms of the pairwise differences or sums of the physical sensor positions \cite{hoctor1990theunifying}. Uniform arrays have a co-array with redundant virtual sensors, which allows the number of physical sensors to be reduced without affecting the number of unique co-array elements. This enables sparse arrays to identify up to $ \mathcal{O}(N^2) $ signal sources using only $ N $ sensors \cite{koochakzadeh2016cramerrao,wang2017coarrays,liu2017cramerrao}. A \emph{contiguous} co-array is often desired, since it maximizes the number of virtual sensors for a given array aperture. The properties of the resulting uniform virtual array model, such as the Vandermonde structure of the virtual steering matrix, can also be leveraged in array processing \cite{pal2010nested,liu2015remarks}.

Typical sparse array designs, such as the \emph{Minimum-Redundancy Array} (MRA) \cite{moffet1968minimumredundancy,hoctor1996arrayredundancy}, seek to maximize the number of contiguous co-array elements for a given number of physical sensors $ N $. The minimum-redundancy property means that no physical array configuration can achieve a larger contiguous co-array. Although optimal, the MRA lacks a closed-form expression for its sensor positions, and quickly becomes impractical to compute, as the search space of the combinatorial optimization problem that needs to be solved grows exponentially with $ N $. Consequently, large sparse arrays have to be constructed by sub-optimal means, yielding low-redundancy, rather than provably minimum-redundancy, array configurations. For instance, smaller (but computable) MRAs can be extended to larger apertures by repeating regular substructures in the array \cite{hoctor1996arrayredundancy}, or by recursively nesting them in a fractal manner \cite{ishiguro1980minimum}. Recent research into such \emph{recursive} or \emph{fractal arrays} \cite{liu2017maximally,yang2018aunified,cohen2019sparsefractal,cohen2020sparse} has also revived interest in \emph{symmetric} array configurations. Symmetry in either the physical or co-array domain can further simplify array design \cite{haupt1994thinned} and calibration \cite{friedlander1991direction}, or be leveraged in detection \cite{xu1994detection}, source localization \cite{roy1989esprit,swindlehurst1992multiple,gao2005ageneralized, wang2013mixed}, and adaptive beamforming \cite[p.~721]{vantrees2002optimum}. \emph{Parametric} arrays are also of great interest, as their sensor positions can be expressed in closed form. This facilitates the simple design and optimization of the array geometry. For example, the redundancy of the array may be minimized to utilize the co-array as efficiently as possible. Notable parametric arrays include the Wichmann \cite{wichmann1963anote,pearson1990analgorithm,linebarger1993difference}, Co-prime \cite{vaidyanathan2011sparsesamplers}, Nested \cite{pal2010nested}, and Super Nested Array \cite{liu2016supernested,liu2016supernestedii}.

Sparse array configurations have been developed mainly for passive sensing, where the \emph{difference co-array} can be exploited if the source signals are incoherent or weakly correlated. Far fewer works consider the \emph{sum co-array}, which is more relevant in active sensing applications, such as radar or microwave and ultrasound imaging, where scatterers may (but need not) be fully coherent \cite{hoctor1990theunifying,hoctor1992highresolution,ahmad2004designandimplementation}. In particular, the design of low-redundancy sparse arrays with overlapping transmitting and receiving elements has not been extensively studied. Some of our recent works have addressed this sum co-array based array design problem by proposing symmetric extensions to existing parametric array configurations that were originally designed to only have a contiguous difference co-array \cite{rajamaki2017sparselinear,rajamaki2018symmetric,rajamaki2020sparselow}. Symmetry can thus provide a simple means of achieving a contiguous sum co-array. However, a unifying analysis and understanding of such symmetric sparse arrays is yet lacking from the literature. The current work attempts to fill this gap.

\subsection{Contributions and organization}
This paper focuses on the design of sparse linear active arrays with a contiguous sum co-array. The main contributions of the paper are twofold. Firstly, we propose a general symmetric sparse linear array design. We establish necessary and sufficient conditions under which the sum and difference co-array of this array are contiguous. We also determine sufficient conditions that greatly simplify array design by allowing for array configurations with a contiguous difference co-array to be leveraged, thanks to the symmetry of the proposed array design. This connects our work to the abundant literature on mostly asymmetric sparse arrays with a contiguous difference co-array \cite{linebarger1993difference,pal2010nested,liu2016supernestedarrays,zheng2019misc}. Moreover, it provides a unifying framework for symmetric configurations relevant to both active and passive sensing \cite{rajamaki2017sparselinear,rajamaki2018symmetric,rajamaki2020sparselow}.

The second main contribution is a detailed study of two specific instances of this symmetric array --- one based on the Nested Array (NA) \cite{pal2010nested}, and the other on the Kl{\o}ve-Mossige basis from additive combinatorics \cite{mossige1981algorithms}. In particular, we clarify the connection between these symmetric arrays and the recently studied \emph{Concatenated Nested Array} (CNA) \cite{rohrbach1937beitrag,rajamaki2017sparselinear} and \emph{Kl{\o}ve Array} (KA) \cite{klove1981class,rajamaki2020sparselow}. We also derive the minimum redundancy parameters for both the CNA and KA. Additionally, we show that the minimum-redundancy symmetric NA reduces to the CNA. Both the CNA and KA can be generated for practically any number of sensors, as their positions have closed-form expressions.

The paper is organized as follows. \cref{sec:preliminaries} introduces the signal model and the considered array figures of merit. \cref{sec:MRA} briefly reviews the MRA and some of its characteristics. \cref{sec:symmetric} then presents the general definition of the proposed symmetric array, and outlines both necessary and sufficient conditions for its sum co-array to be contiguous. In \cref{sec:generators}, we study two special cases of this array, and derive their minimum-redundancy parameters. Finally, we compare the discussed array configurations numerically in \cref{sec:numerical}, before concluding the paper and discussing future work in \cref{sec:conclusions}.

\subsection{Notation}
We denote matrices by bold uppercase, vectors by bold lowercase, and scalars by unbolded letters. Sets are denoted by calligraphic letters. The set of integers from $ a \in \mathbb{Z}$ to $ c\in\mathbb{Z} $ in steps of $ b\in\mathbb{N}_+ $ is denoted $ \{a:b:c\}= \{a,a+b,a+2b,\ldots,c\}$. Shorthand $ \{a:c\}$ denotes $\{a:1:c\} $. The sum of two sets is defined as the set of pairwise sums of elements, i.e., $ \mathcal{A}+\mathcal{B}\triangleq \{a+b\ |\ a\in\mathcal{A};b\in\mathcal{B} \} $. The sum of a set and a scalar is a special case, where either summand is a set with a single element. Similar definitions hold for the difference set. The cardinality of a set $ \mathcal{A} $ is denoted by $ |\mathcal{A}| $. The rounding operator $ \lceil \cdot \rfloor $ quantizes the scalar real argument to the closest integer. Similarly, the ceil operator $ \lceil \cdot \rceil $ yields the smallest integer larger than the argument, and the floor operator $ \lfloor\cdot \rfloor $ yields the largest integer smaller than the argument.

\section{Preliminaries}\label{sec:preliminaries}
In this section, we briefly review the active sensing and sum co-array models. We then define the considered array figure of merit, which are summarized in \cref{tab:fom}.

\subsection{Signal model}\label{sec:signal_model}
Consider a linear array of $ N $ transmitting and receiving sensors, whose normalized positions are given by the set of non-negative integers $ \mathcal{D}\!=\!\{d_n\}_{n=1}^{N}\subseteq\mathbb{N} $. The first sensor of the array is located at $ d_1 = 0 $, and the last sensor at $ d_N=L$, where  $ L=\max \mathcal{D}$ is the (normalized) array aperture. This array is used to actively sense $ K $ far field scatterers with reflectivities $\{\gamma_k\}_{k=1}^K\subseteq \mathbb{C}$ in 
azimuth directions $ \{\varphi_k\}_{k=1}^K \subseteq[-\pi/2, \pi/2]$. Each transmitter illuminates the scattering scene using narrowband radiation in a sequential or simultaneous (orthogonal MIMO) manner \cite{hoctor1992highresolution,li2007mimoradar}. The reflectivities are assumed fixed during the coherence time of the scene, which may consist of one or several snapshots (pulses). The received signal after a single snapshot and matched filtering is then \cite{boudaher2015sparsitybased}
\begin{align}
\bm{x} = (\bm{A}\odot \bm{A})\bm{\gamma} + \bm{n}, \label{eq:z}
\end{align}
where  $ \odot $ denotes the Khatri-Rao (columnwise Kronecker) product, $ \bm{A}\!\in\!\mathbb{C}^{N\times K} $ is the array steering matrix, $ \bm{\gamma}\!=\![\gamma_1,\ldots,\gamma_K]^\txt{T}\!\in\!\mathbb{C}^K$ is the scattering coefficient vector, and $ \bm{n}\!\in\!\mathbb{C}^{N^2} $ is a receiver noise vector following a zero-mean white complex circularly symmetric normal distribution. A typical array processing task is to estimate parameters $\{ \varphi_k, \gamma_k\}_{k=1}^K$, or some functions thereof, from the measurements $ \bm{x} $.

\subsection{Sum co-array}\label{sec:co-array}
The effective steering matrix in \eqref{eq:z} is given by $ \bm{A} \odot \bm{A}$. Assuming ideal omnidirectional sensors, we have
\begin{align*}
	[\bm{A}\odot\bm{A}]_{(n-1)N+m,k}=\exp({j2\pi (d_n+d_m)\delta\sin\varphi_k}),
\end{align*}
 where $ \delta $ is the unit inter-sensor spacing in carrier wavelengths (typically $ \delta=1/2 $). Consequently, the entries of $\bm{A}\odot\bm{A}$ are supported on a virtual array, known as the \emph{sum co-array}, which consists of the pairwise sums of the physical element locations.
\begin{definition}[Sum co-array]\label{def:sca}
	The virtual element positions of the sum co-array of physical array $ \mathcal{D}$ are given by the set
	\begin{align}
	\mathcal{D}_\Sigma \triangleq \mathcal{D}+\mathcal{D} = \{d_n+d_m\ |\ d_n,d_m\in\mathcal{D} \}.\label{eq:sca}
	\end{align}
\end{definition}
The relevance of the sum co-array is that it may have up to $N(N+1)/2 $ unique elements, which is vastly more than the number of physical sensors $ N $. This implies that $\mathcal{O}(N^2)$ coherent scatterers can be identified from \eqref{eq:z}, provided the set of physical sensor positions $ \mathcal{D} $ is judiciously designed. 

The sum co-array is \emph{uniform} or \emph{contiguous}, if it equals a virtual \emph{Uniform linear array} (ULA) of aperture $ 2L=2\max\mathcal{D} $.
\begin{definition}[Contiguous sum co-array]
	The sum co-array of $ \mathcal{D} $ is contiguous if $ \mathcal{D}+\mathcal{D} = \{0:2\max\mathcal{D} \} $.
\end{definition}

A contiguous co-array is desirable for two main reasons. Firstly, it maximizes the number of unique co-array elements for a given physical aperture. Second, it facilitates the use of many array processing algorithms designed for uniform arrays. For example, \emph{co-array MUSIC} \cite{pal2010nested,liu2015remarks} leverages the resulting Vandermonde structure to resolve more sources than sensors unambiguously in the infinite snapshot regime.

A closely related concept to the sum co-array is that of the difference co-array \cite{haubrich1968array,hoctor1990theunifying}. Defined as the set of pairwise element position differences, the difference co-array mainly emerges in passive sensing applications, where the incoherence of the source signals is leveraged. Other assumptions give rise to more exotic co-arrays models, such as the \emph{difference of the sum} co-array\footnote{Up to $\mathcal{O}(N^4)$ \emph{incoherent} scatterers can be resolved by utilizing the second-order statistics of \eqref{eq:z} and the difference of the sum co-array \cite{boudaher2015sparsitybased}.} \cite{chen2008minimumredundancymimo,weng2011nonuniform}, and the \emph{union} of the sum and difference co-array \cite{wang2017doaestimation,si2019improved}, which are not considered herein.

\makeatletter \renewcommand{\@IEEEsectpunct}{\ \,}\makeatother
\subsection{Array figures of merit}\label{sec:fom}
\begin{table}[]
	\centering
	\caption{Frequently used symbols. The set of sensor positions is denoted by $\mathcal{D}$,  the array aperture by $ L=\max\mathcal{D} $, and the number of sensors by $ N=|\mathcal{D}| $.}\label{tab:fom}
	\resizebox{1\linewidth}{!}{%
		\begin{tabular}{c|c|c}
			Symbol&Explanation&Value range\\
			\hline
			$ |\mathcal{D}_\Sigma| $&Number of total DoFs&$ \{N:2L+1\}$\\
			$ H $&Number of contiguous DoFs&$ \{1:|\mathcal{D}_\Sigma|\}$\\
			$ R $&Redundancy&$ [1,\infty)$\\
			$ S(d) $&$ d $-spacing multiplicity&$ \{0:\min(N-1,L-d+1)\}$\\
			$ R_\infty $&Asymptotic redundancy&$ [1,\infty)$\\
			$ F_\infty $&Asymptotic co-array filling ratio&$ [0,1]$\\
			\end{tabular}
	}
\end{table}

\subsubsection{Degrees of freedom (DoF).}\label{sec:dof}
The number of unique elements in the sum co-array $|\mathcal{D}_\Sigma|$ is often referred to as the total number of DoFs. Similarly, the number of virtual elements in the largest contiguous subarray contained in the sum co-array is called the number of contiguous DoFs.
\begin{definition}[Number of contiguous DoFs]\label{def:H}
	The number of contiguous DoFs in the sum co-array of $\mathcal{D}$ is
	\begin{align}
		H\triangleq \argmax _{h,s\in\mathbb{N}}\big\{ h\ |\ s+\{0:h-1\}\subseteq \mathcal{D}+\mathcal{D}\big\}. \label{eq:H}
	\end{align}
\end{definition}
If the offset $ s$ is zero, then $ H $ equals the position of the first hole in the sum co-array. Moreover, if the sum co-array is contiguous, then $ H=2L+1 $, where $ L $ is the array aperture. 

\subsubsection{Redundancy,} \label{sec:R}
$ R $, quantifies the multiplicity of the co-array elements. A non-redundant array achieves $ R\!=\!1 $, whereas $R \!>\!1 $ holds for a redundant array. 
\begin{definition}[Redundancy] \label{def:R}
	The redundancy of an $ N $ sensor array with $ H$ contiguous sum co-array elements is
	\begin{align}
		R \triangleq \frac{N(N+1)/2}{H}. \label{eq:R}
	\end{align}
\end{definition}
The numerator of $ R $ is the maximum number of unique pairwise sums generated by $ N $ numbers. The denominator is given by \eqref{def:H}. \cref{def:R} is essentially \emph{Moffet's} definition of (difference co-array) redundancy \cite{moffet1968minimumredundancy} adapted to the sum co-array. It also coincides with \emph{Hoctor and Kassam's} definition \cite{hoctor1996arrayredundancy}, when the sum co-array is contiguous, i.e., $ H\!=\!2L\!+\!1 $, where $ 2L+1 $ is the aperture of the sum co-array.

\subsubsection{The $d$-spacing multiplicity,} \label{sec:fom_S}
$ S(d) $, enumerates the number of inter-element spacings of a displacement $ d $ in the array \cite{rajamaki2018sparseactive}. For linear arrays, $ S(d) $ simplifies to the \emph{weight} or \emph{multiplicity function} \cite{pal2010nested,rajamaki2017comparison} of the difference co-array (when $ d\geq 1 $).
\begin{definition}[$ d $-spacing multiplicity \cite{rajamaki2018sparseactive}] \label{def:S}
	The multiplicity of inter-sensor displacement $ d \geq 1$ in a linear array $ \mathcal{D} $ is
	\begin{align}
	S(d) \triangleq \frac{1}{2}\sum_{d_n\in \mathcal{D}}\sum_{d_m \in \mathcal{D}}\mathbbm{1}(|d_n-d_m| = d).\label{eq:S}
	\end{align}
\end{definition}
It is easy to verify that $ 0\leq S(d)\leq \min(N-1,L-d+1)$, where $d\in\mathbb{N}_+ $. If the difference co-array is contiguous, then $ S(d)\geq 1$ for $1 \leq d\leq L $.

Typically, a low value for $ S(d) $ is desired for small $ d $, as sensors that are closer to each other tend to interact more strongly and exhibit coupling \cite{friedlander1991direction}. Consequently, the severity of mutual coupling effects deteriorating the array performance may be reduced by decreasing $ S(d) $ \cite{boudaher2016mutualcoupling,liu2016supernested,liu2017hourglass,zheng2019misc}. This simplifies array design, but has its limitations. Specifically $ S(d) $ does not take into account important factors, such as the individual element gain patterns and the mounting platform, as well as the scan angle and the uniformity of the array \cite[Ch.~8]{balanis2016antenna}. Since treating such effects in a mathematically tractable way is challenging, proxies like the \emph{number of unit spacings} $ S(1) $ are sometimes considered instead for simplicity.

\subsubsection{Asymptotic and relative quantities.}
The discussed figures of merit are generally functions of $ N $, $ H $, or $ L $, i.e., they depend on the size of the array. A convenient quantity that holds in the limit of large arrays is the \emph{asymptotic redundancy}
\begin{align}
	R_\infty \triangleq \lim_{N\to\infty} R= \lim_{N\to\infty} \frac{N^2}{2H}.\label{eq:R_inf}
\end{align}
Another one is the \emph{asymptotic sum co-array filling ratio}
	\begin{align}
		F_\infty \triangleq \lim_{N\to\infty}\frac{H}{2L+1},\label{eq:F_inf}
	\end{align}
which satisfies $ 0\!\leq\!F_\infty\!\leq\!1 $, as $ H\!\leq\!|\mathcal{D}_\Sigma|\!\leq\!2L+1 $. If the sum co-array is contiguous, then $ F_\infty\!=\!1 $. Note that the limits in \eqref{eq:R_inf} and \eqref{eq:F_inf} may equivalently be taken with respect to $ L $ or $ H $.

In many cases, we wish to evaluate the relative performance of a given array with respect to a reference array configuration of choice -- henceforth denoted by superscript ``ref''. Of particular interest are the three ratios 
\begin{align*}
	\frac{H}{H^\txt{ref}},\ \frac{N}{N^\txt{ref}}\ \txt{and}\  \frac{L}{L^\txt{ref}},
\end{align*}
or their asymptotic values, when either $ H,N $ or $ L $ is equal for both arrays and approaches infinity. \cref{tab:relative} shows that these asymptotic ratios can be expressed using \eqref{eq:R_inf} and \eqref{eq:F_inf}, provided the respective limits exist. For example, the second row of the first column in \cref{tab:relative} is interpreted as 
\begin{align*}
	\lim_{N\to\infty}\frac{H(N)}{H^\txt{ref}(N)} =\frac{R_\infty^\txt{ref}}{R_\infty},
\end{align*}
where the right-hand-side follows by simple manipulations from \eqref{eq:R}. The expressions in \cref{tab:relative} simplify greatly if both arrays have a contiguous sum co-array, since $ F_\infty=F_\infty^\txt{ref}=1 $.
\begin{table}[]
	\centering
	\caption{Asymptotic ratios of the number of contiguous DoFs $ H $, sensors $ N $, and array aperture $ L $. The variable approaching infinity is assumed equal for both arrays.}\label{tab:relative}
		\begin{tabular}{c|c|c|c}
			&$ H/H^\txt{ref} $&$ N/N^\txt{ref} $&$ L/L^\txt{ref} $\\
			\hline
			$ H\to\infty $&$ 1 $&$ \sqrt{R_\infty/R_\infty^\txt{ref}} $&$F_\infty^\txt{ref}/F_\infty  $\\
			$ N\to\infty $&$R_\infty^\txt{ref}/R_\infty$&$1$&$(R_\infty^\txt{ref}/R_\infty) (F_\infty^\txt{ref}/F_\infty )$\\
			$ L\to\infty $&$ F_\infty/F_\infty^\txt{ref} $&$ \sqrt{R_\infty/R_\infty^\txt{ref}} \sqrt{F_\infty/F_\infty^\txt{ref}} $&$1$\\
		\end{tabular}
\end{table}

\makeatletter \renewcommand{\@IEEEsectpunct}{:\ \,}\makeatother

\section{Minimum-Redundancy Array}\label{sec:MRA}

In this section, we present the sparse array design problem solved by the \emph{Minimum-redundancy array} (MRA). We then review some properties of the MRA, and briefly discuss an extension that is computationally easier to find.

The MRA is defined as the solution to either of two slightly different optimization problems, depending on if the sum co-array is constrained to be contiguous or not \cite{moffet1968minimumredundancy}. 
\begin{definition}[Minimum-redundancy array (MRA)] \label{def:MRA}
	The general Minimum-redundancy array (MRA) solves
	\begin{opteq}
		\underset{\mathcal{D}\subseteq\mathbb{N}; h\in\mathbb{N}}{\txt{maximize}}&\ h\nonumber \\
		\txt{subject to}&\ |\mathcal{D}|=N\ \txt{and}\ \mathcal{D}+\mathcal{D} \supseteq \{0:h-1\}. \label{p:MRA_general}
	\end{opteq}
The restricted MRA (R-MRA) is given by the solution to \eqref{p:MRA_general} with the extra constraint $ h=2\max\mathcal{D}+1$.
\end{definition}
The general MRA minimizes the redundancy $ R $, subject to a given number of sensors $ N $ and offset $s=0$ in \eqref{eq:H}. By \cref{def:R}, this  is equivalent to maximizing $ H $, which by \cref{def:H} reduces to \eqref{p:MRA_general}. In contrast, the R-MRA constrains the sum co-array to be contiguous, and therefore maximizes the array aperture. A more generic definition of the general MRA is possible by including offset $s\in\mathbb{N}$ as an optimization variable in \eqref{p:MRA_general}. Note that the R-MRA implies that $s=0$, regardless of the definition of the MRA. Finding general MRAs with $s\neq0$ is an open problem that is left for future work.

MRAs can, but need not, be restricted \cite{kohonen2014addition,kohonen2014meet}. For example, two of the three MRAs with $ N=8 $ sensors in \cref{tab:MRA_example} are restricted. On the other hand, none of the general MRAs with $ N=11 $ sensors are restricted ($ H=47 $ in the general case, respectively, $ H=45 $ in the restricted case) \cite{kohonen2014addition,kohonen2014meet}. The ``sum MRAs'' in \cref{def:MRA} are equivalent to \emph{extremal additive 2-bases}, which have been extensively studied in number theory \cite{rohrbach1937beitrag,riddell1978someextremal,mossige1981algorithms,klove1981class,kohonen2014meet}, similarly to \emph{difference bases} \cite{leech1956ontherepresentation,wichmann1963anote} corresponding ``difference MRAs'' \cite{moffet1968minimumredundancy}. Solving \eqref{p:MRA_general} is nevertheless challenging, since the size of the search space grows exponentially with the number of elements $ N $. Consequently, MRAs are currently only known for $ N\leq 25 $ \cite{kohonen2014addition} and R-MRAs for $ N\leq48 $ \cite{kohonen2015early}. 

The MRA can also be defined for a fixed aperture. E.g., the restricted MRA would minimize $ N $ subject to a contiguous co-array of length $ 2L+1 $. The difference MRA following this definition is referred to as the \emph{sparse ruler} \cite{shakeri2012directionofarrival}. For a given $ N $, several such rulers of different length may exist. However, we exclusively consider the MRA of \cref{def:MRA} henceforth.

\subsection{Key properties}\label{sec:MRA_properties}
The lack of a closed-form expression for the sensor positions of the MRA make its properties difficult to analyze  precisely. Nevertheless, it is straightforward to see that a linear array with a contiguous sum co-array necessarily contains two sensors in each end of the array, as shown by the following lemma.
\begin{lemma}[Necessary sensors]\label{thm:N_nec}
	Let $ N\geq 2 $. If $ \mathcal{D} $ has a contiguous sum co-array, then $ \mathcal{D}\supseteq \{0,1,L-1,L\} $. If $ \mathcal{D} $ has a contiguous difference co-array, then $ \mathcal{D}\supseteq \{0,1,L\} $. 
\end{lemma}
\begin{proof}
	Clearly, $ \mathcal{D}+\mathcal{D} \supseteq \{0,1,2L-1,2L\} $ if and only if $ \mathcal{D}\supseteq \{0,1,L-1,L\} $. Similarly, $ \mathcal{D}-\mathcal{D} \supseteq \{0,1,L-1,L\} $ if and only if $\mathcal{D}\supseteq \{0,1,L\} $ or $\mathcal{D}\supseteq \{0,L-1,L\}$. We may write $ \mathcal{D}\supseteq \{0,1,L\} $ without loss of generality, since any $ \mathcal{D}\supseteq \{0,L-1,L\} $ can be mirrored to satisfy $ L- \mathcal{D}\supseteq \{0,1,L\}$.
\end{proof}
\cref{thm:N_nec} implies that any array with a contiguous sum co-array and $ N\geq 4 $ sensors has at least two sensor pairs separated by a unit inter-element spacing, i.e., $  S(1)\geq2 $. \cref{thm:N_nec} also suggests that the R-MRA achieves redundancy $ R=1 $ in only two cases. These arrays are called \emph{perfect arrays}.
\begin{corollary}[Perfect arrays]
	The only two unit redundancy R-MRAs are $ \{0\} $ and $ \{0,1\} $. All other R-MRAs are redundant.
\end{corollary}
\begin{proof}
	By \cref{thm:N_nec}, the element at position $ L $ of the sum co-array can be represented in at least two ways, namely $ L=L+0=(L-1)+1 $. Consequently, if $ N\geq 4 $, then $ R>1 $ must hold. The R-MRAs for $ N\leq 3 $ are $ \{0\}, \{0,1\} $, and $ \{0,1,2\} $. Only the first two of these satisfy $ R=1 $.
\end{proof}
Generally, the redundancy of the MRA is unknown. However, the asymptotic redundancy can be bounded as follows.
\begin{thm}[Asymptotic redundancy of MRA \cite{yu2015anew,kohonen2017animproved,klove1981class,yu2009upper}]\label{thm:R_inf_MRA}
	The asymptotic redundancy of the MRA satisfies
	\begin{align*}
	\txt{General MRA: }& 1.090 < R_\infty  < 1.730\\
	\txt{Restricted MRA: }& 1.190 < R_\infty  < 1.917.
	\end{align*}
\end{thm}
\begin{proof}
	In the general case, the lower bound $ R_\infty\geq \frac{1}{0.917} > 1.090 $ follows directly from \cite[Theorem 1.1]{yu2015anew}, and the upper bound $R_\infty\leq \frac{147}{85} < 1.730 $ from \cite[Eq.~(1)]{kohonen2017animproved}. Similarly, in the restricted case $ R_\infty\geq \frac{11}{7+\sqrt{5}} > 1.190 $ \cite[Theorem 1.2]{yu2009upper}, and $R_\infty\leq \frac{23}{12} < 1.917 $ \cite[Theorem, p.~177]{klove1981class}.
\end{proof}
\cref{thm:R_inf_MRA} suggests that the R-MRA may be more redundant than the general MRA that does not constrain the sum co-array to be contiguous. The restricted definition of the MRA is nevertheless more widely adopted than the general one for the reasons listed in \cref{sec:co-array}. We note that difference bases or MRAs can also be of the general or restricted type \cite{leech1956ontherepresentation,moffet1968minimumredundancy}. Difference MRAs are typically less redundant than sum MRAs due to the commutativity of the sum, i.e., $ a+b=b+a $, but $ a-b\neq b-a $.

\subsection{Unique solution with fewest closely spaced sensors}\label{sec:MRA_unique}
Problem \eqref{p:MRA_general} may have several solutions for a given $ N $, which means that the MRA is not necessarily unique \cite{kohonen2014addition,kohonen2014meet}. In order to guarantee uniqueness, we introduce a secondary optimization criterion. In particular, we consider the MRA with the \emph{fewest closely spaced sensors}. This MRA is found by minimizing a weighted sum of $ d $-spacing multiplicities (see \cref{def:S}) among the solutions to \eqref{p:MRA_general}, which is equivalent to subtracting a regularizing term from the objective function. This regularizer $\varsigma\geq0 $ can be defined as, for example,
\begin{align}
\varsigma\triangleq \sum_{d=1}^LS(d)10^{-d(\lfloor \log L\rfloor + 1)}, \label{eq:varsigma}
\end{align}
where $ L \in\mathbb{N}_+$ is the largest aperture of the considered solutions. Consequently, any two solutions to the unregularized problem, say $ \mathcal{D}_a $ and $ \mathcal{D}_b $, satisfy $ \varsigma_a > \varsigma_b$, if and only if $ S_a(n)>S_b(n) $ and $ S_a(d)=S_b(d)$ for all $ 1\leq d< n$. In words: \eqref{eq:varsigma} promotes large sensor displacements by prioritizing the value of $ S(1) $, then $ S(2) $, then $ S(3) $, etc. For example, \cref{tab:MRA_example} shows two R-MRAs with equal $ S(1) $ and $ S(2) $, but different $ S(3) $. The R-MRA with the smaller $S(3) $, and therefore lower value of $ \varsigma $,  is preferred.
\begin{table}[]
	\centering
	\caption{MRAs with $ N=8 $ sensors \cite{kohonen2014addition}. The bottom MRA has the fewest unit spacings. Of the two restricted MRAs, the first has fewer closely spaced sensors and hence a lower $\varsigma$.}\label{tab:MRA_example}
	\resizebox{1\linewidth}{!}{%
	\begin{tabular}{c|c|c|c|c|c}
		Configuration&Restricted&$ S(1) $&$ S(2) $&$ S(3) $&$ \varsigma $\\
		\hline
		$ \{0,1,2,5,8,11,12,13\} $&\ding{51}&$ 4 $&$ 2 $&$ 3 $&$ 0.040203\ldots $\\
		$ \{0,1,3,4,9,10,12,13\} $&\ding{51}&$ 4 $&$ 2 $&$ 4 $&$ 0.040204\ldots $\\
		$ \{0,1,3,5,7,8,17,18\} $&\ding{55}&$ 3 $&$ 3 $&$ 2 $&$ 0.030302\ldots $\\
	\end{tabular}%
}
\end{table}

\subsection{Symmetry and the Reduced Redundancy Array}
Many of the currently known sum MRAs are actually symmetric \cite{kohonen2014addition,kohonen2014meet}. In fact, there exist at least one symmetric R-MRA for each $ N\leq 48 $. Moreover, the R-MRA with the fewest closely spaced sensors (lowest $\varsigma$) turns out to be symmetric for $ N\leq 48 $. Indeed, symmetry seems to arise naturally from the additive problem structure (cf. \cite[Section~4.5.3]{kohonen2015exact}). Note that difference MRAs are generally asymmetric \cite{moffet1968minimumredundancy,ishiguro1980minimum}. 

Imposing symmetry on the array design problem has the main advantage of reducing the size of the search space \cite{mossige1981algorithms,haupt1994thinned}. In case of the MRA, this can be achieved by adding constraint $ \mathcal{D}= \max \mathcal{D}-\mathcal{D} $ to \eqref{p:MRA_general}. Unfortunately, the search space of this symmetric MRA still scales exponentially with $ N $. Fortunately, another characteristic of the MRA may readily be exploited. Namely, MRAs tend to have a sparse mid section consisting of uniformly spaced elements \cite{kohonen2014meet}. The \emph{reduced redundancy array} (RRA) extends the aperture of the MRA by adding extra elements to this uniform mid section \cite{ishiguro1980minimum,hoctor1996arrayredundancy}.
\begin{definition}[Reduced redundancy array (RRA) \cite{hoctor1996arrayredundancy}]
	The sensor positions of the RRA for a given MRA are given by
	\begin{align*}
	\mathcal{D}_\txt{RRA} \triangleq \mathcal{P} \cup (\mathcal{M}+\max \mathcal{P})\cup (\mathcal{S}+\max \mathcal{P}+\max \mathcal{M}),
	\end{align*}
	where $ \mathcal{P} $ is the prefix and $ \mathcal{S} $ the suffix of the MRA, and 
	\begin{align*}
	\mathcal{M} =	\{0:M:(N-|\mathcal{P}|-|\mathcal{S}|+1)M \}
	\end{align*}
	is the mid subarray, with inter-element spacing $ M \in\mathbb{N}_+$.
\end{definition}

The prefix $ \mathcal{P} $, suffix $ \mathcal{S} $, and the inter-element spacing $ M $ of $ \mathcal{M} $ are determined by the \emph{generator} MRA, i.e., the MRA that is extended. For example, an MRA with $ N=7 $ sensors is 
\begin{align*}
\underbrace{\{0, 1, 2, 5, 8, 9, 10\}}_{\mathcal{D}_\txt{MRA}}  = \underbrace{\{0, 1, 2\} }_{\mathcal{P}}\cup\underbrace{\{2:3:8\}}_{\mathcal{M}+2}\cup \underbrace{\{8, 9, 10\}}_{\mathcal{S}+8}.
\end{align*}
Note that $ \mathcal{S}= \max \mathcal{P}-\mathcal{P} $ holds, if the MRAs is symmetric as above.  The RRA also has a contiguous sum co-array, if the MRA is restricted.

The RRA has a low redundancy when $ |\mathcal{D}_\txt{RRA}|\approx|\mathcal{D}_\txt{MRA}| $. However, the redundancy of the RRA approaches infinity as $ |\mathcal{D}_\txt{RRA}|$ grows, since the aperture of the RRA only increases linearly with the number of sensors. Consequently, we will next consider a general class of symmetric arrays which scale better and admit a solution in polynomial time, provided the design space is constrained judiciously. In particular, we will show that this class of symmetric arrays naturally extends many of the established array configurations -- designed for passive sensing -- to the active sensing setting. 

\section{Symmetric array --- general design and conditions for contiguous co-array} \label{sec:symmetric}
In this section, we establish a general framework for symmetric arrays with a contiguous sum (and difference) co-array. The proposed \emph{symmetric array with generator $\mathcal{G}$} (S-$\mathcal{G} $) is constructed by taking the union of a generator array\footnote{This terminology is adopted from the literature on fractal arrays \cite{werner2003anoverview,cohen2019sparsefractal}.} $\mathcal{G}$ and its mirror image shifted by a non-negative integer $ \lambda $. The generator, which is the basic construction block of the symmetric array, can be any array configuration of choice.
\begin{definition}[Symmetric array with generator $ \mathcal{G} $ (S-$\mathcal{G}$)]\label{def:SA}
	The sensor positions of the S-$\mathcal{G}$ are given by
	\begin{align}
	\mathcal{D}_{\txt{S-}\mathcal{G}} \triangleq\mathcal{G}\cup (\max \mathcal{G}-\mathcal{G}+\lambda),\label{eq:SA}
	\end{align}
	where $\mathcal{G}$ is a generator array and $\lambda\in \mathbb{N}$ is a shift parameter\footnote{Note that $ \lambda \leq 0 $ is equivalent to considering the mirrored generator $ \max{\mathcal{G}}-{\mathcal{G} }$ for $ \lambda \geq 0 $. Therefore, we may set $ \lambda \geq 0$ without loss of generality.}.
\end{definition}
\cref{fig:S-G_array} shows a schematic of the S-$\mathcal{G} $. Note that the number of sensors satisfies $ |\mathcal{G}|\leq N\leq 2|\mathcal{G}|$,
depending on the overlap between $ \mathcal{G} $ and $L- \mathcal{G} $. The aperture $ L $ of the S-$\mathcal{G}$ is
\begin{align}
	L= \max \mathcal{G} + \lambda. \label{eq:SA_L}
\end{align}
The exact properties of the S-$\mathcal{G}$ are determined by the particular choice of $ \mathcal{G} $ and $\lambda$. Specifically, the generator array $ \mathcal{G} $ influences how easily the symmetrized array S-$\mathcal{G}$ can be optimized to yield a large contiguous sum co-array. For example, a generator with a simple nested structure makes a convenient choice in this regard, as we will demonstrate in \cref{sec:generators}. Next, however, we establish necessary and sufficient conditions that any $\mathcal{G}$ and $\lambda$ must fulfill for S-$\mathcal{G} $ to have a contiguous sum co-array. These general conditions are later utilized when considering specific choices for $ \mathcal{G} $ in \cref{sec:generators}.

\subsection{Necessary and sufficient conditions for contiguous co-array}
\cref{fig:S-G_coarray} illustrates the difference co-array of the S-$\mathcal{G}$, which is composed of the difference co-array and shifted sum co-arrays of the generator $\mathcal{G}$. By symmetry, the sum and difference co-array of the S-$ \mathcal{G} $ are equivalent up to a shift. This fact, along with \eqref{eq:SA}, allows us to express the necessary and sufficient condition for the contiguity of both co-arrays in terms of the generator $\mathcal{G}$ and shift parameter $ \lambda $. Moreover, we may conveniently decompose the condition into two simpler subconditions, as shown by the following theorem.
\begin{thm}[Conditions for contiguous co-array] \label{thm:SA_coarray_c}
	The sum (and difference) co-array of the S-$\mathcal{G}$ is contiguous if and only~if
	\begin{enumerate}[label=(C\arabic*),leftmargin=1.0cm]
		\item $ (\mathcal{G}-\mathcal{G})\cup (\mathcal{G}+\mathcal{G}-L)\cup  (L-(\mathcal{G}+\mathcal{G}))\supseteq \{0:\max\mathcal{G}\}$\label{c:gmg}
	\item $ \mathcal{G}+\mathcal{G}\supseteq \{0:\lambda-1\}, $ \label{c:gpg}
	\end{enumerate}
	where $ L $ is the array aperture given by \eqref{eq:SA_L}.
\end{thm}
\begin{proof}
	By symmetry of the physical array, the sum co-array is contiguous if and only if the difference co-array is contiguous (e.g., see \cite[Lemma~1]{rajamaki2018sparseactive}), that is,
	\begin{align*}
		\mathcal{D}_{\txt{S-}\mathcal{G}}+\mathcal{D}_{\txt{S-}\mathcal{G}} =\{0:2L\} \iff  \mathcal{D}_{\txt{S-}\mathcal{G}}-\mathcal{D}_{\txt{S-}\mathcal{G}} =\{-L:L\}.
	\end{align*}
	By symmetry of the difference co-array, this is equivalent to requiring that 
	$ \mathcal{D}_{\txt{S-}\mathcal{G}}\!-\!\mathcal{D}_{\txt{S-}\mathcal{G}}\!\supseteq\!\{0:L\} $, or using \eqref{eq:SA} that
	\begin{align*}
	\mathcal{D}_{\txt{S-}\mathcal{G}}-\mathcal{D}_{\txt{S-}\mathcal{G}} &=\mathcal{G}\cup (L-\mathcal{G})-\mathcal{G}\cup (L-\mathcal{G})\\
	&=	(\mathcal{G}-\mathcal{G})\cup(\mathcal{G}+\mathcal{G}-L)\cup(L-(\mathcal{G}+\mathcal{G}))\\
	&\supseteq \{0:L\}.
	\end{align*}
	Conditions \ref{c:gpg} and \ref{c:gmg} then directly follow from \eqref{eq:SA_L}.
\end{proof}
Note that \ref{c:gpg} may also be reformulated as $ \lambda\leq H$, where $ H $ denotes the position of the first hole in the sum co-array of $\mathcal{G}$, per \eqref{eq:H}. Since $ H\leq 2\max\mathcal{G}+1 $, the sum co-array of $\mathcal{D}_{\txt{S-}\mathcal{G}}$ is contiguous only if $ \lambda $ satisfies
\begin{align*}
\lambda \leq 2\max \mathcal{G}+1.
\end{align*}

\subsection{Sufficient conditions for contiguous co-array}
It is also instructive to consider some simple sufficient conditions for satisfying \ref{c:gmg} and \ref{c:gpg}, as outlined in the following two corollaries to \cref{thm:SA_coarray_c}.
\begin{corollary}[Sufficient conditions for \ref{c:gmg}]\label{thm:c_suff_gmg}
	Condition~\ref{c:gmg} is satisfied if either of the following holds:
	\begin{enumerate}[label=(\roman*)]
		\item$ \mathcal{G} $ has a contiguous difference co-array.\label{gmg_1}
		\item $ \mathcal{G} $ has a contiguous sum co-array and $\lambda\leq \max\mathcal{G}+1$.\label{gmg_2}
	\end{enumerate}	
\end{corollary}
\begin{proof}
	Firstly, if the difference co-array of $ \mathcal{G}$ is contiguous, then $ \mathcal{G} - \mathcal{G}\supseteq \{0:\max\mathcal{G}\} $ holds by definition, implying \ref{c:gmg}. Secondly, if the sum co-array is contiguous, then $ \mathcal{G} + \mathcal{G}\supseteq \{0:2\max\mathcal{G}\} $ holds, which implies that
	\begin{align*}
			\mathcal{G}+\mathcal{G}-L &= \{-L:\max\mathcal{G}-\lambda\}\\
		L-(\mathcal{G}+\mathcal{G}) &= \{-\max\mathcal{G}+\lambda:L\}.
	\end{align*}
	 If $ \lambda\leq \max\mathcal{G} +1$, then the union of these two sets and $\mathcal{G}-\mathcal{G}$ covers $ \{-L:L\}\supseteq\{0:\max\mathcal{G}\} $, implying \ref{c:gmg}.
\end{proof}
\begin{corollary}[Sufficient conditions for \ref{c:gpg}]\label{thm:c_suff_gpg}
	Condition \ref{c:gpg} is satisfied if any of the following hold:
	\begin{enumerate}[label=(\roman*)]
		\item $ \lambda \leq 1 $ \label{gpg_1}
		\item $ \lambda \leq 3 $ and $ |\mathcal{G}|\geq 2 $ \label{gpg_2}
		\item $ \lambda \leq 2\max\mathcal{G}+1 $ and $ \mathcal{G} $ has a contiguous sum co-array.\label{gpg_3}
	\end{enumerate}
\end{corollary}
\begin{proof}
	Firstly, $ \lambda \leq 1 $ implies \ref{c:gpg}, since $ \mathcal{G}\supseteq\{0\} $. Secondly, if $ |\mathcal{G}| \geq 2$, then $\mathcal{G}\supseteq\{0,1\}$ holds by \cref{thm:N_nec}. Consequently, $ |\mathcal{G}| \geq 2 $ and $ \lambda \leq 3  \implies$ \ref{c:gpg}. Thirdly, if $ \mathcal{G} $ has a contiguous sum co-array and $ \lambda\leq  2\max\mathcal{G}+1$, then $ \mathcal{G}\!+\!\mathcal{G}\!=\!\{0:2\max\mathcal{G}\}\supseteq\{0:\lambda\!-\!1\}$ holds by definition.
\end{proof}
We note that if $ |\mathcal{G}|\geq 2 $, then $ \lambda\leq \max\mathcal{G}+2 $ is actually sufficient in \cref{gmg_2} of \cref{thm:c_suff_gmg} (cf. \cref{thm:N_nec}). This is also sufficient for satisfying \ref{c:gpg} by \cref{gpg_3} of \cref{thm:c_suff_gpg}. \cref{gmg_1} of \cref{thm:c_suff_gmg} is of special interest, since it states that any $\mathcal{G}$ with a contiguous difference co-array satisfies \ref{c:gmg}. This greatly simplifies array design, as shown in the next section, where we develop two arrays leveraging this property.

\begin{figure*}[]
	\centering
	\subfigure[Physical array and constituent sub-arrays]{\includegraphics[width=0.3\textwidth]{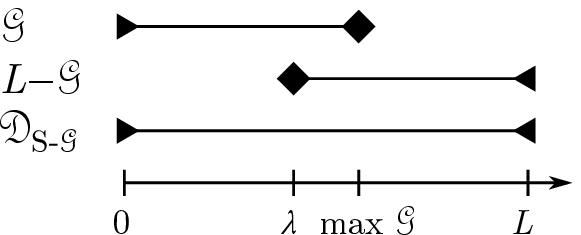}\label{fig:S-G_array}}
	\hfil
	\subfigure[Difference co-array (same as sum co-array but shifted by $ -L $) and constituent sub-co-arrays]{\includegraphics[width=0.6\textwidth]{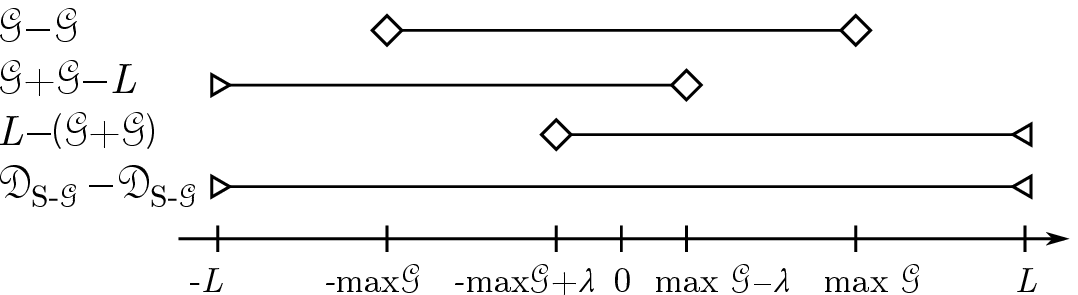}\label{fig:S-G_coarray}}	
	\caption{The symmetric array S-$ \mathcal{G} $ in (a) consists of the union of a generator array $ \mathcal{G} $ and its mirror image shifted by $\lambda\geq0$. The difference co-array of the S-$ \mathcal{G} $ in (b) can be decomposed into the union of the difference co-array of $ \mathcal{G} $, and two shifted copies of the sum co-array of $\mathcal{G}$ (of which one is also mirrored). Due to symmetry, the difference and sum co-array of the S-$ \mathcal{G} $ are equivalent up to a shift of $ L $, and contiguous only if $ \lambda\leq2\max\mathcal{G}+1 $.}
\end{figure*}

\section{Low-redundancy symmetric array designs using parametric generators}\label{sec:generators}

Similarly to the R-MRA in \eqref{p:MRA_general}, we wish to find the S-$\mathcal{G} $ with maximal aperture 
satisfying the conditions in \cref{thm:SA_coarray_c}. Given a number of elements $ N $, and a class of generator arrays $ \mathscr{G} $, such that $\mathcal{G} \in\mathscr{G} $, this minimum-redundancy  S-$\mathcal{G} $ design is found by solving the following optimization problem:
\begin{opteq}
	\underset{\mathcal{G}\in\mathscr{G}, \lambda \in \mathbb{N}}{\text{maximize}}&\ \max \mathcal{G}+\lambda\nonumber \\
	\text{subject to}&\ |\mathcal{G}\cup (\max \mathcal{G}-\mathcal{G}+\lambda)|=N\nonumber\\
	&\ \txt{\ref{c:gmg} and \ref{c:gpg}}. \label{p:SA}
\end{opteq}
In general, \eqref{p:SA} is a non-convex problem, whose difficulty depends on the choice of $\mathscr{G}$. Solving \eqref{p:SA} may therefore require a grid search over $ \lambda $ and the elements of $ \mathscr{G} $, which can have exponential complexity at worst. At best, however, a solution can be found in polynomial time, or even in closed form.

We will now focus on a family of choices for $\mathscr{G}$, such that each $ \mathcal{G}\in\mathscr{G} $ has the following two convenient properties:
\begin{enumerate}[label=(\alph*)]
	\item $\mathcal{G}$ has a contiguous difference co-array\label{i:G_diff_cont}
	\item $\mathcal{G}$ has a closed-form expression for its aperture.\label{i:G_param}
\end{enumerate} 
\cref{i:G_diff_cont} greatly simplifies \eqref{p:SA} by directly satisfying condition \ref{c:gmg}, whereas \cref{i:G_param} enables the straightforward optimization of the array parameters. Condition \ref{c:gpg} is typically easy to satisfy for an array with these two properties. This is the case with many sparse array configurations in the literature, such as the Nested \cite{pal2010nested} and Wichmann Array \cite{wichmann1963anote,pearson1990analgorithm,linebarger1993difference}. Constructing a symmetric array using a generator with a contiguous difference co-array is thus a practical way to synthesize a contiguous sum co-array.

\subsection{Symmetric Nested Array (S-NA)}
The \emph{Nested Array} (NA) \cite{pal2010nested} is a natural choice for a generator satisfying the previously outlined properties~\ref{i:G_diff_cont} and \ref{i:G_param}. The NA has a simple structure consisting of the union of a dense ULA, $\mathcal{D}_1$, and sparse ULA sub-array, $\mathcal{D}_2$, as shown in \cref{fig:G_NA}. When used as the generator $\mathcal{G}$ in \eqref{eq:SA}, the NA yields the \emph{Symmetric Nested Array} (S-NA), defined as follows:
\begin{definition}[Symmetric Nested Array (S-NA)] \label{def:S-NA}
	The sensor positions of the S-NA are given by \eqref{eq:SA}, where 
	\begin{align*}
	\mathcal{G} = \mathcal{D}_1 \cup (\mathcal{D}_2+N_1),
	\end{align*}
	with $\mathcal{D}_1 = \{0:N_1-1\};$  $\mathcal{D}_2=\{0:N_1+1:(N_2-1)(N_1+1)\}$; and array parameters $ N_1,N_2 \in\mathbb{N}$.
\end{definition}
Special cases of the S-NA have also been previously proposed. For example, \cite[Definition~2]{liu2018optimizing} considered the case $ \lambda=0 $ for improving the robustness of the NA to sensor failures. Next we briefly discuss a special case that is more relevant from the minimum-redundancy point of view.

\subsubsection{Concatenated Nested Array (CNA)}
The S-NA coincides with a restricted additive 2-basis studied by Rohrbach in the 1930's \cite[Satz~2]{rohrbach1937beitrag}, when the shift parameter $ \lambda $ satisfies
\begin{align}
	\lambda= (N_1+1)k+N_1, \label{eq:lambda_CNA}
\end{align}
where $ k\in \{0:N_2\} $. Unaware of Rohrbach's early result, we later derived the same configuration based on the NA and called it the \emph{Concatenated Nested Array} (CNA) \cite{rajamaki2017sparselinear}.
\begin{definition}[Concatenated Nested Array (CNA) \cite{rajamaki2017sparselinear}]\label{def:CNA}
	The sensor positions of the CNA are given by 
	\begin{align*}
	\mathcal{D}_\txt{CNA} \triangleq \mathcal{D}_1 \cup (\mathcal{D}_2 +N_1) \cup (\mathcal{D}_1+N_2(N_1+1)),
	\end{align*}
	where $ N_1,N_2 \in\mathbb{N}$, and $ \mathcal{D}_1,\mathcal{D}_2$ follow \cref{def:S-NA}.
\end{definition}

The CNA is illustrated in \cref{fig:CNA}. When $N_1 = 0$ or $N_2 \in\{0,1\}$, the array degenerates to the ULA. In the interesting case when $ N_1+N_2\geq 1 $, the aperture of the CNA is\cite{rajamaki2017sparselinear}
\begin{align}
L &= (N_1+1)(N_2+1)-2. \label{eq:L_CNA}
\end{align}
Furthermore, the number of sensors is \cite{rajamaki2017sparselinear}
\begin{align}
N = 
\begin{cases}
N_1,&\txt{if } N_2=0\\
2N_1+N_2,& \txt{otherwise,}
\end{cases} \label{eq:N_CNA}
\end{align}
and the number of unit spacings evaluates to
\begin{align}
S(1) &= \begin{cases}
N_1-1,&\txt{if } N_2=0\\
N_2-1,&\txt{if } N_1=0\\
2N_1,&\txt{otherwise}.
\end{cases}\label{eq:S1_CNA}
\end{align}

\begin{figure}[]
	\centering
	\subfigure[Nested Array (NA) \cite{pal2010nested}]{\includegraphics[width=.6\linewidth]{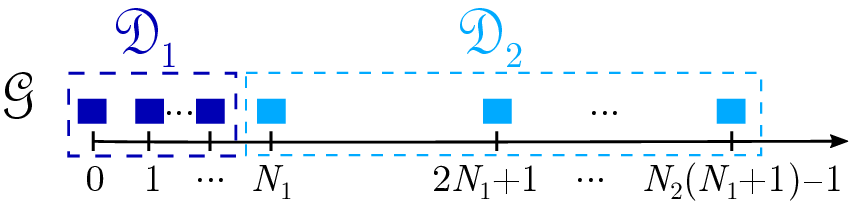}\label{fig:G_NA}}
	\subfigure[Symmetric Nested Array (S-NA)]{\includegraphics[width=0.85\linewidth]{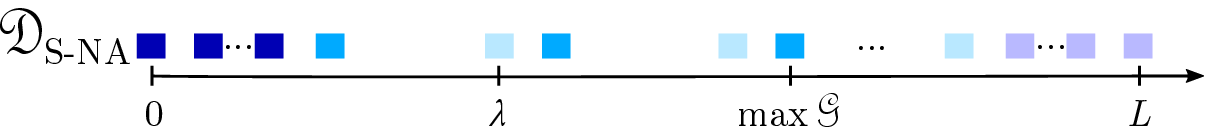}\label{fig:S-NA}}
	\subfigure[Concatenated Nested Array (CNA) \cite{rohrbach1937beitrag,rajamaki2017sparselinear}]{\includegraphics[width=1\linewidth]{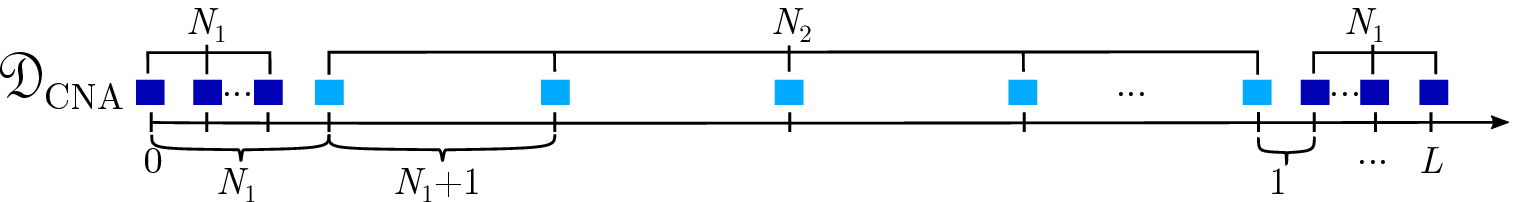}\label{fig:CNA}}
	\caption{The (a) NA generator and shift parameter $\lambda$ define the (b) S-NA. The S-NA reduces to the (c) CNA, when $\lambda$ follows \eqref{eq:lambda_CNA}. The minimum-redundancy S-NA solving \eqref{p:SA} is a CNA.}
\end{figure}

\subsubsection{Minimum-redundancy solution} \label{sec:S-NA_MRA}
The S-NA solving \eqref{p:SA} is actually a CNA. This follows from the fact that any S-NA that has a contiguous sum co-array, but is not a CNA, has redundant sensors, due to the reduced number of overlapping sensors between $\mathcal{G}$ and its shifted mirror image.
\begin{proposition}\label{thm:S-NA-CNA}
	The S-NA solving \eqref{p:SA} is a CNA.
\end{proposition}
\begin{proof}
	We write \ref{c:gmg} and \ref{c:gpg} as a constraint on $\lambda$, which we then show implies the proposition under reparametrization.
	
	Firstly, \ref{c:gmg} holds by \cref{gmg_1} of \cref{thm:c_suff_gmg}, since the NA has a contiguous difference co-array \cite{pal2010nested}. Secondly, the location of the first hole in the sum co-array of the NA yields that \ref{c:gpg} is satisfied if and only if 
	\begin{align}
		\lambda \leq 
		\begin{cases}
		2(N_1+N_2)-1,& \txt{ if } N_1=0 \txt{ or } N_2=0\\
		N_2(N_1+1)+N_1,& \txt{ otherwise}.
		\end{cases} \label{eq:lambda_max_S_NA}
	\end{align}
	If $0\leq \lambda \leq N_1-1$, then a CNA with the same number of sensors ($ N=2|\mathcal{G}|-2 $), but a larger aperture ($ L=\max\mathcal{G}+\lambda $) is achieved by instead setting $ \lambda= \max\mathcal{G}-N_1 $. Otherwise, it is straightforward to verify that the S-NA either reduces to the CNA, or a CNA can be constructed with the same $ N $ but a larger $ L $, by satisfying \eqref{eq:lambda_max_S_NA} with equality.
\end{proof}

By \cref{thm:S-NA-CNA}, problem \eqref{p:SA} simplifies to maximizing the aperture of the CNA for a given number of sensors:
\begin{opteq}
	\underset{N_1\in\mathbb{N},N_2\in\mathbb{N}_+}{\text{maximize}}&\ N_1N_2+N_1+N_2\ \ \text{s.t.}\ \ 2N_1+N_2=N. \label{p:CNA}
\end{opteq}
Problem \eqref{p:CNA} actually admits a closed-form solution \cite{rajamaki2017sparselinear}. In fact, we may state the following more general result for any two-variable integer program similar to \eqref{p:CNA}.
\begin{lemma} \label{thm:general_opt}
	Let $ f$ be a concave function. The solution to
	\begin{opteq}
		\underset{x,y\in\mathbb{N}}{\text{maximize}}\ f(x,y)\ \text{subject to}\ g(x)=y \label{p:gen}
	\end{opteq}
	is given by $ x = \lceil z \rfloor + k$ and $ y=g(x) $, where $ z $ solves
	\begin{align*}
		\underset{z\in\mathbb{R}_+}{\text{maximize}}&\ f(z,g(z)),
	\end{align*}
	and $ |k|$ is the smallest non-negative integer satisfying $ g( \lceil z \rfloor + k) \in \mathbb{N}$, where $ k\in \{-\lceil z \rfloor,-\lceil z \rfloor+1,\ldots,-1\}\cup \mathbb{N}$.
\end{lemma}
\begin{proof}
	Since $ f(z,g(z)) $ is concave, $ x=\lceil z\rfloor $ maximizes $ f(x,g(x)) $ among all $ x\in\mathbb{N}$. This is a global optimum of \eqref{p:gen}, if $ g(\lceil z\rfloor)\in\mathbb{N} $. Generally, the optimal solution can be expressed as $ x= \lceil z\rfloor + k$, where $ k \in \{-\lceil z \rfloor:-1\}\cup \mathbb{N}$. By concavity of $ f $, the smallest $ |k| $ satisfying $ g(\lceil z\rfloor + k)\in\mathbb{N} $ then yields the global optimum of \eqref{p:gen}.
\end{proof}

\cref{thm:general_opt} is useful for solving two-variable integer programs similar to \eqref{p:CNA} in closed-form. Such optimization problems often arise in, e.g., sparse array design \cite{rajamaki2018symmetric,rajamaki2020sparselow}. In our case, \cref{thm:general_opt} allows expressing the minimum-redundancy parameters of the CNA directly in terms of $ N $ as follows\footnote{An equivalent result is given in \cite[Eq.~(7) and (8)]{rohrbach1937beitrag} and \cite[Eq.~(13)]{rajamaki2017sparselinear}, but in a more inconvenient format due to the use of the rounding operator.}.
\begin{thm}[Minimum-redundancy parameters of CNA] \label{thm:param_CNA}
	The parameters of the CNA solving \eqref{p:CNA} are
	\begin{align}
	N_1 &= (N-\alpha)/4 \label{eq:N1_CNA_opt}\\
	N_2 &= (N+\alpha)/2, \label{eq:N2_CNA_opt}
	\end{align}
	where $ N=4m+k, m\in \mathbb{N}, k \in \{0:3\}$, and
	\begin{align}
	\alpha = (k+1)\bmod 4 -1.\label{eq:alpha_CNA}
	\end{align}
\end{thm}
\begin{proof}
	By \cref{thm:general_opt}, the optimal solution to \eqref{p:CNA} is given by $ N_1=\lceil (N-1)/4 \rfloor $ and $ N_2 = N-2N_1 $ \cite{rajamaki2017sparselinear}. Since any $ N\in\mathbb{N} $ may be expressed as $ N=4m+k$, where $m\in\mathbb{N} $ and $ k \in \{0:3\}$, we have
	\begin{align*}
	N_1 = \Big\lceil m+\frac{k-1}{4}\Big\rfloor 
	= \begin{cases}
	m,& k \in \{0,1,2\}\\
	m+1,& k = 3.
	\end{cases}
	\end{align*}
	Since $ m = (N-k)/4 $, we have $ N_1 = m = (N-k)/4 $, when $ k \in \{0,1,2\} $. Similarly, when $ k=3 $, we have $ N_1 = m+1 = (N-k+4)/4 =(N+1)/4$, which yields \eqref{eq:N1_CNA_opt} and \eqref{eq:alpha_CNA}. Substituting \eqref{eq:N1_CNA_opt} into \eqref{eq:N_CNA} then yields \eqref{eq:N2_CNA_opt}.
\end{proof}

By \cref{thm:param_CNA}, the properties of the minimum-redundancy CNA can also be written compactly as follows.
\begin{corollary}[Properties of minimum-redundancy CNA]
	The aperture $ L $, number of sensors $ N $, and number of unit spacings $ S(1) $ of the CNA solving \eqref{p:CNA} are
	\begin{align*}
	L &= (N^2+6N-7)/8-\beta\\
	N &= 2\sqrt{2}\sqrt{L+2+\beta}-3\\
	S(1) &= (N-\alpha)/2,
	\end{align*}
	where $ \beta= (\alpha -1)^2/8 $ and $ \alpha $ is given by \eqref{eq:alpha_CNA}.
\end{corollary}
\begin{proof}
	This follows from \cref{thm:param_CNA}, and \eqref{eq:L_CNA}--\eqref{eq:S1_CNA}. 
\end{proof}

\subsection{Symmetric Kl{\o}ve-Mossige array (S-KMA)} \label{sec:KA}
In the 1980's, \emph{Kl{\o}ve and Mossige} proposed an additive 2-basis with some interesting properties \cite{mossige1981algorithms}. In particular, the basis has a contiguous difference co-array (see Appendix~\ref{a:KMA_diff_coarray}) and a low asymptotic redundancy ($ R_\infty\!=\!1.75 $), despite having a non-contiguous sum co-array (see Appendix~\ref{a:S-KMA_lambda}). We call this construction the \emph{Kl{\o}ve-Mossige Array} (KMA). As shown in \cref{fig:G_KMA}, the KMA contains a CNA, and can therefore be interpreted as another extension of the NA. However, selecting the KMA as the generator in \eqref{eq:SA} yields the novel \emph{Symmetric Kl{\o}ve-Mossige Array} (S-KMA), shown in \cref{fig:S-KMA}.
\begin{definition}[Symmetric Kl{\o}ve-Mossige Array (S-KMA)]\label{def:S-KMA}
	The sensor positions of the S-KMA are given by \eqref{eq:SA}, where
	\begin{align*}
	\mathcal{G} &=\mathcal{D}_\txt{CNA} \cup (\mathcal{D}_3+2\max \mathcal{D}_\txt{CNA}+1),\\
	\mathcal{D}_3&= \{0:N_1:N_1^2\}+\bigcup_{i=1}^{N_3}\{(i-1)(N_1^2+\max \mathcal{D}_\txt{CNA}+1)\},
	\end{align*}
	$\mathcal{D}_\txt{CNA}$ follows \cref{def:CNA}, and parameters\footnote{The S-KMA is undefined for $ N_1=N_2=0 $, as $\mathcal{D}_\txt{CNA} = \emptyset$. Consequently, we will not consider this case further. However, we note that \cref{def:S-KMA} is easily modified so that the S-KMA degenerates to a ULA even in this case.} $ N_1, N_2,N_3\!\in\!\mathbb{N} $.
\end{definition}

\subsubsection{Kl{\o}ve array}
The structure of the S-KMA simplifies substantially when the shift parameter $ \lambda $ is of the form
\begin{align}
\lambda=
2\max\mathcal{D}_\txt{CNA}+1+(\max\mathcal{D}_\txt{CNA}+N_1^2)k, \label{eq:lambda_KA}
\end{align}
where $ k\!\in\!\{0\!:\!N_3\}$. In fact, this S-KMA coincides with the \emph{Kl{\o}ve array} (KA), which is based on a class of restricted additive 2-bases proposed by \emph{Kl{\o}ve} in the context of additive combinatorics \cite{klove1981class} (see also \cite{rajamaki2020sparselow}).
\begin{definition}[Kl{\o}ve array (KA) \cite{klove1981class}]
	The sensor positions of the KA with parameters $N_1,N_2,N_3\!\in\!\mathbb{N}$ are given by
	\begin{align*}
	\mathcal{D}_\txt{KA} \triangleq&\  \mathcal{D}_\txt{CNA} \cup (\mathcal{D}_3+2\max \mathcal{D}_\txt{CNA}+1)\nonumber\\ 
	&\cup (\mathcal{D}_\txt{CNA}+(N_3+2)\max \mathcal{D}_\txt{CNA}+N_3(N_1^2+1)+1),
	\end{align*}
 where $\mathcal{D}_\txt{CNA}$ follows \cref{def:CNA}, and $\mathcal{D}_3$ \cref{def:S-KMA}.
\end{definition}

\cref{fig:KA} illustrates the KA, which consists of two CNAs connected by a sparse mid-section consisting of $ N_3  $ widely separated and sub-sampled ULAs with $  N_1+1 $ elements each. The KA reduces to the CNA when $ N_1 =0 $ or $ N_2 =0 $. When $ N_2\geq 1 $, the aperture $ L $, and the number of sensors $ N $ of the KA evaluate to \cite[Lemma 3]{klove1981class}:
\begin{align}
L&=(N_1+1)(N_3(N_1+N_2)+3N_2+3)-5 \label{eq:L_KA}\\
N&=2(2N_1+N_2)+N_3(N_1+1).\label{eq:N_KA}
\end{align}
Furthermore, the number of unit spacings $ S(1) $ is \cite{rajamaki2020sparselow}
\begin{align}
S(1) &=
\begin{cases}
N_3+1,& \txt{if } N_1 = 0 \txt{ and } N_2  = 1\\
2(N_2-1),& \txt{if } N_1 = 0 \txt{ and } N_2\geq 2\\
N_3+4,& \txt{if } N_1 = 1\\
4N_1,& \txt{otherwise.}
\end{cases}\label{eq:S1_KA}
\end{align}
\begin{figure*}[]
	\centering
	\subfigure[Kl{\o}ve-Mossige array (KMA) \cite{mossige1981algorithms}]{\includegraphics[width=0.85\textwidth]{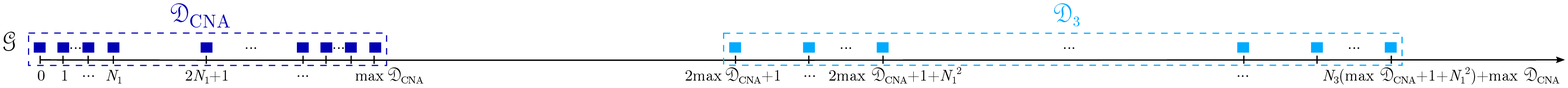}\label{fig:G_KMA}}
	\hfil
	\subfigure[Symmetric Kl{\o}ve-Mossige array (S-KMA)]{\includegraphics[width=1\textwidth]{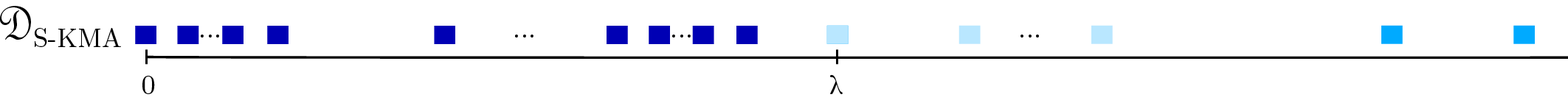}\label{fig:S-KMA}}
	\hfil
	\subfigure[Kl{\o}ve array (KA) \cite{klove1981class,rajamaki2020sparselow}]{\includegraphics[width=.9\textwidth]{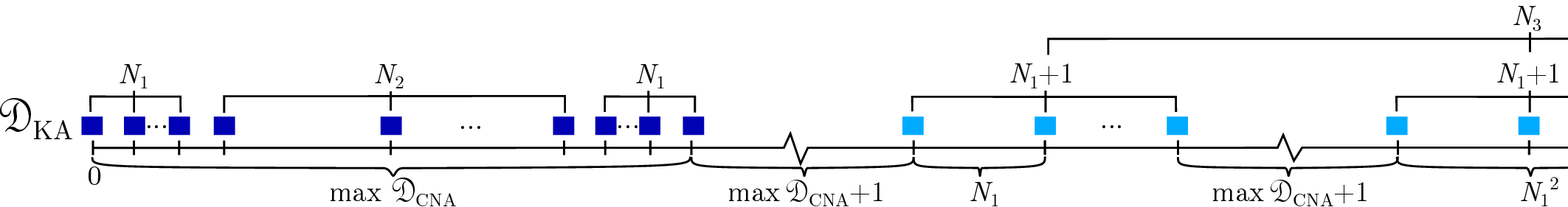}\label{fig:KA}}
	\caption{The Kl{\o}ve-Mossige generator in (a) has a contiguous difference co-array, but a non-contiguous sum co-array. The S-KMA in (b) reduces to the KA in (c), when $\lambda$ follows \eqref{eq:lambda_KA}. The KA has a contiguous sum co-array and may be interpreted as a generalization of the CNA.}
\end{figure*}

\subsubsection{Minimum-redundancy solution}
The simple structure of the KMA suggests that the minimum redundancy S-KMA is also a KA. Intuitively, any S-KMA that is not a KA is unnecessarily redundant, since increasing $ \lambda $ to its maximum value $ H $ yields a KA (cf. Appendix~\ref{a:S-KMA_lambda}). However, proving that the S-KMA solving \eqref{p:SA} is a KA is not as simple as in the case of the S-NA and CNA (see \cref{thm:S-NA-CNA} in \cref{sec:S-NA_MRA}). In particular, the S-KMA cannot easily be modified into a KA with an equal or larger aperture without affecting the number of sensors $ N $. A formal proof is therefore still an open problem.

Nevertheless, we may derive the minimum redundancy parameters of the KA. The \emph{minimum-redundancy Kl{\o}ve array} (KA$ _R $) maximizes the aperture for a given $ N $, which is equivalent to solving the following optimization problem:
\begin{opteq}
	\underset{N_1,N_3\in\mathbb{N}, N_2\in\mathbb{N}_+}{\text{maximize}}&\ 3N_1(N_2+1)+3N_2+N_3(N_1+N_2)(N_1+1)\nonumber \\
	\text{subject to}&\ 2(2N_1+N_2)+N_3(N_1+1)=N. \label{p:KA_R}
\end{opteq}
Problem \eqref{p:KA_R} is an integer program with three unknowns. This problem is challenging to solve in closed form for \emph{any} $ N $, in contrast to \eqref{p:CNA}, which only has two integer unknowns (cf. \cref{thm:param_CNA}). Nevertheless, \emph{some} $ N $ readily yield a closed-form solution to \eqref{p:KA_R}, as shown by the following theorem.
\begin{thm}[Minimum-redundancy parameters of KA] \label{thm:param_KA_R}
	The parameters of the KA solving \eqref{p:KA_R} are 
	\begin{align}
	N_1 &= (N+3)/23 \label{eq:N1_relaxed}\\
	N_2 &= 5(N+3)/23\label{eq:N2_relaxed}\\
	N_3 &= (9N-42)/(N+26),\label{eq:N3_relaxed}
	\end{align}
	when $ N\!\in\!\{20,43,66,112,250\} $.
\end{thm}
\begin{proof}\let\qed\relax 
	Under certain conditions, we may obtain a solution to \eqref{p:KA_R} by considering a relaxed problem. Specifically, solving \eqref{eq:N_KA} for $ N_3 $, substituting the result into \eqref{eq:L_KA}, and
	relaxing $ N_1,N_2\in\mathbb{R} $ leads to the following concave quadratic program:
	\begin{align*}
	\underset{N_1,N_2\in\mathbb{R}}{\text{maximize}}&\ (N_1+N_2)(N\!+\!3)-3N_1N_2-4N_1^2-2N_2^2.
	\end{align*}
	At the critical point of the objective function we have
	\begin{align*}
	\partial L/\partial N_1&=N-8N_1-3N_2+3 = 0\\
	\partial L/\partial N_2&=N-3N_1-4N_2+3 = 0.
	\end{align*}
	Solving these equations for $ N_1$ and $N_2 $ yields \eqref{eq:N1_relaxed} and \eqref{eq:N2_relaxed}, which when substituted into \eqref{eq:N_KA} yields \eqref{eq:N3_relaxed}. These are also solutions to \eqref{p:KA_R} if:
	\begin{enumerate}[label=(\roman*)]
		\item $ N_1 \in\mathbb{N}$, which holds when $ N= 23k-3, k\in\mathbb{N}_+ $, i.e., $ N =20,43,66,89,112,135,158,181,204,227,250,\ldots $
		\item $ N_2 \in \mathbb{N} $, which holds when $ N_1\in\mathbb{N}$, as $ N_2 = 5 N_1 \in \mathbb{N} $
		\item $ N_3 \in\mathbb{N} $, which holds when $ N=(26l+42)/(9-l)\in \mathbb{N} $ and $ l\in\{0:8\} $, i.e., $ N=20,43,66,112,250 $.
	\end{enumerate}
	The only integer-valued $ N $ satisfying all three conditions are $ N\!=\!20,43,66,112,250 $, as stated.
\end{proof}

The KA in \cref{thm:param_KA_R} also has the following properties:
\begin{corollary}[Properties of minimum-redundancy KA]
	The aperture $ L $, number of sensors $ N $, and number of unit spacings $ S(1) $ of the KA solving \eqref{p:KA_R} are
	\begin{align*}
	L &= (3N^2+18N-19)/23 \label{eq:L_opt_KA}\\
	N &= \sqrt{23/3}\sqrt{L+2}-3\\
	S(1) &=
	\begin{cases}
	(N+22)/6& N_1 = 1\\
	4(N+3)/23 & N_1\geq 2,
	\end{cases}
	\end{align*}
	when $ N\!\in\!\{20,43,66,112,250\}. $
\end{corollary}
\begin{proof}
	This follows from \cref{thm:param_KA_R} and \eqref{eq:L_KA}--\eqref{eq:S1_KA}.
\end{proof}

The minimum-redundancy KA achieves the asymptotic redundancy $ R_\infty\!=\!23/12 $, as shown in the following proposition, which is a reformulation of \cite[Theorem, p.~177]{klove1981class}.
\begin{proposition}[Asymptotic redundancy of KA \cite{klove1981class}]\label{thm:R_inf_KA_R}
	The asymptotic redundancy of the solution to \eqref{p:KA_R} is
	\begin{align*}
	R_\infty = 23/12 < 1.9167.
	\end{align*}
\end{proposition}
\begin{proof}
Let $ N=23k+9$, where $k\in\mathbb{N} $. A feasible KA that is equivalent to the minimum-redundancy KA when $ N\to\infty $ is then given by the choice of parameters $ N_1 = (N-9)/23 $, $ N_2 = 5N_1 $, and $ N_3=9$. Substitution of these parameters into \eqref{eq:L_KA} yields $ L= 3N^2/23 +\mathcal{O}(N) $, i.e., $ R_\infty  = 23/12 $.
\end{proof}

\subsubsection{Polynomial time grid search}
Although solving \eqref{p:KA_R} in closed form for any $ N $ is challenging, we can nevertheless obtain the solution in at most $\mathcal{O}(N\log N)$ objective function evaluations. This follows from the fact that the feasible set of \eqref{p:KA_R} only has $\mathcal{O}(N\log N)$ or fewer elements.
\begin{proposition}[Cardinality of feasible set in \eqref{p:KA_R}]\label{thm:KA_complexity}
	The cardinality of the feasible set in \eqref{p:KA_R} is at most $ \mathcal{O}(N\log N)$.
\end{proposition}
\begin{proof}\let\qed\relax 
	We may verify from \eqref{eq:N_KA} that $ 0 \leq N_1\leq (N-2)/4 $ and $ 0 \leq N_3\leq (N-4N_1)/(N_1+1) $. Consequently, the number of grid points that need to be checked is
	\begin{align*}
	V = \sum_{N_1=0}^{\lfloor\frac{N-2}{4}\rfloor}\Bigg\lfloor\frac{N-4N_1}{N_1+1}\Bigg\rfloor+1 \leq  \int_{0}^{\frac{N+2}{4}} \frac{N-3x+5/2}{x+1/2}dx.
	\end{align*}
	The upper bound follows from ignoring the floor operations, and substituting $ N_1\!=\!x\!-\!1/2$, where $ x\!\in\!\mathbb{R} $, to account for the rectangular integration implied by the sum. Finally,
	\[
		V\!\leq\!(N\!+\!4)\log(N/2\!+\!2)\!-\!3(N\!+\!2)/4\!=\!\mathcal{O}(N\log N)
	\]
	follows from integration by parts.
\end{proof}

\cref{alg:grid} summarizes a simple grid search that finds the solution to \eqref{p:KA_R} for any $ N $ in $ \mathcal{O}(N\log N)$ steps, as implied by \cref{thm:KA_complexity}. We iterate over $ N_1 $ and $ N_3 $, because this choice yields the least upper bound on the number of grid points that need to be checked\footnote{Selecting $ N_1 $ and $ N_2$, or $ N_2$ and $ N_3 $, yields $ \mathcal{O}(N^2) $ points.}. Since the solution to \eqref{p:KA_R} is not necessarily unique, we select the KA$_R $ with the fewest closely spaced sensors, similarly to the MRA in \cref{sec:MRA_unique}. Note that computing the regularizer $ \varsigma $ requires $ \mathcal{O}(N^2) $ floating point operations (flops), whereas evaluating the aperture $ L $ only requires $\mathcal{O}(1)$ flops. Consequently, the time complexity of finding the KA$_R $ with the fewest closely spaced sensors is\footnote{Actually, $ \varsigma $ only needs to be computed when $ L\!=\!L^\star $ in \cref{alg:grid}.} $\mathcal{O}(N^2)$, whereas finding any KA$ _R $, that is, solving \eqref{p:KA_R} in general, has a worst case complexity of $ \mathcal{O}(N\log N) $.

\begin{algorithm}[]
	\caption{Minimum-redundancy parameters of Kl{\o}ve Array} \label{alg:grid}
	\begin{algorithmic}[1]		
		\Procedure{KA$_R $}{$N$}
		\State $ f\gets -\infty $\Comment{initialize objective function value}
		\For{$ N_1 \in \{0:\lfloor (N-2)/4 \rfloor\} $}
		\For{$ N_3\in \{0:\lfloor (N-4N_1)/(N_1+1) \rfloor\} $}
		\State $ N_2 \gets (N-(N_1+1)N_3)/2-2N_1 $
		\If{${N_2 \bmod 1\!=\!0}$}\Comment{$ N_2 $ valid if integer}
		\State Compute $ L $ and $ \varsigma $ using \eqref{eq:L_KA} and \eqref{eq:varsigma} \label{line:varsigma}
		\If{$L-\varsigma > f$} \Comment{better solution found}
		\State $ f\gets L-\varsigma$\Comment{update objective fcn.}
		\For{$ i\in\{1:3\} $} $ N_i^\star\gets N_i $\EndFor
		\EndIf
		\EndIf
		\EndFor
		\EndFor
		\State \Return $N_1^\star,N_2^\star,N_3^\star$ \Comment{optimal array parameters}
		\EndProcedure
	\end{algorithmic}
\end{algorithm}

In a recent related work, we also developed a KA with a constraint on the number of unit spacings \cite{rajamaki2020sparselow}. This \emph{constant unit spacing Kl{\o}ve Array} (KA$_S $) achieves $ S(1)=8 $ for any $ N $ at the expense of a slight increase in redundancy (asymptotic redundancy $ R_\infty=2 $). The minimum-redundancy parameters of the KA$_S $ are found in closed form by application of \cref{thm:general_opt}, similarly to the CNA (cf. \cite[Theorem~1]{rajamaki2020sparselow}).

\subsection{Other generator choices}
Naturally, other choices for the generator $\mathcal{G}$ may yield alternative low-redundancy symmetric array configurations. For example, the Wichmann generator with $ \lambda = 0 $ yields the \emph{Interleaved Wichmann Array} (IWA) \cite{rajamaki2018symmetric}. This array satisfies \ref{c:gmg} by \cref{gmg_1} of \cref{thm:c_suff_gmg} and \ref{c:gpg} by \cref{gpg_1} of \cref{thm:c_suff_gpg}. The IWA has the same asymptotic redundancy as the CNA, but fewer unit spacings. Although finding the minimum-redundancy parameters of the more general \emph{Symmetric Wichmann Array} (S-WA) in closed-form is cumbersome, numerical optimization of $\lambda$ and the other array parameters can slightly improve the non-asymptotic redundancy of the S-WA compared to the IWA. Nevertheless, neither the IWA nor S-WA are considered further in this work, since the KA achieves both lower $ R $ and $ S(1) $. 

Numerical experiments suggest that some prominent array configurations, such as the \emph{Super Nested Array} \cite{liu2016supernested} or  \emph{Co-prime Array} \cite{vaidyanathan2011sparsesamplers}, are not as well suited as generators $\mathcal{G}$ from the redundancy point of view (mirroring $\mathcal{G}$ does not help either). However, several other generators, both with and without contiguous difference co-arrays, remain to be explored in future work. For example, the difference MRA \cite{moffet1968minimumredundancy} and minimum-hole array \cite{bloom1977applications,vertatschitsch1986nonredundant} are interesting candidates.

\section{Performance evaluation of array designs} \label{sec:numerical}
In this section, we compare the sparse array configurations presented in \cref{sec:MRA,sec:generators} in terms of the array figures of merit in \cref{sec:fom}. We also demonstrate the proposed sparse array configurations' capability of resolving more scatterers than sensors in an active sensing scenario. Further numerical results can be found in the companion paper \cite{rajamaki2020sparselow}.

\subsection{Comparison of array figures of merit}\label{sec:comparison}
\cref{tab:summary,tab:asymptotic} summarize the key properties and figures of merit of the discussed sparse array configurations. The parameters of each configuration is chosen such that they maximize the number of contiguous DoFs $ H$ in \eqref{eq:H}. The optimal parameters of the NA and KMA are found similarly to those of the CNA and KA (cf. \cref{thm:param_CNA,thm:param_KA_R})\footnote{Derivations are straightforward and details are omitted for brevity.}. \cref{fig:arrays} illustrates the array configurations for $ N=24  $ sensors. The RRA is omitted since the MRA is known in this case. Note that the sum and difference co-arrays of the symmetric arrays are contiguous and merely shifted copies of each other.

\begin{table*}[]
	\resizebox{1\linewidth}{!}{
		\begin{threeparttable}[t]
			\centering
			\caption{Key properties of considered sparse array configurations. The arrays below the dashed line have a contiguous sum co-array. The symmetric arrays have equivalent sum and difference co-arrays.}\label{tab:summary}
			\begin{tabular}{c|c|c|c|c|c|c}
				Array configuration&Symmetric&Contiguous DoFs\tnote{a}, $ H $&Total DoFs, $ |\mathcal{D}_\Sigma| $&Aperture, $ L $& No. of sensors, $ N $&No. of unit spacings, $ S(1)$\\
				\hline
				General Minimum-Redundancy Array (MRA)&
				no&
				n/a&
				n/a&
				$ \geq (H-1)/2$ and $\leq  H $&
				n/a&
				$ \geq 1 $\\
				Nested Array (NA) \cite{pal2010nested}& 
				no&
				$ (N^2+2N-4)/4$&
				$ H+N/2-1 $& 
				$ (N^2+4N)/4$&
				$ \sqrt{4L+5}-1 $&
				$ N/2 $\\
				Kl{\o}ve-Mossige Array (KMA) \cite{mossige1981algorithms}&
				no&
				$ (2N^2+8N+1)/7$&
				$ H+6N/7+\mathcal{O}(1)$&
				$ (11N^2+16N-61)/49$&
				$ (7\sqrt{11L+15}-8)/11 $&
				$2(N+2)/7$\\
			   \hdashline
				Restricted Minimum-Redundancy Array (R-MRA)&
				no\tnote{b}&
				n/a&
				$ H $&
				$ (H-1)/2 $&
				n/a&
				$ \geq 2 $\\
				Reduced-Redundancy Array (RRA) \cite{hoctor1996arrayredundancy}&
				yes&
				$ 30N-706 $&
				$ H $&
				$ 15N-353 $&
				$ (L+353)/15 $&
				$ 10 $\\
				Concatenated Nested Array (CNA) \cite{rajamaki2017sparselinear}&
				yes&
				$ (N^2+6N-3)/4$&
				$ H $&
				$ (N^2+6N-7)/8$&
				$  2\sqrt{2}\sqrt{L+2}-3$&
				$ N/2 - 1/2 $\\
				Constant unit spacing Kl{\o}ve Array (KA$_S $) \cite{rajamaki2020sparselow}&
				yes&
				$(N^2+10N-83)/4 $&
				$ H$&
				$ (N^2+10N-87)/8$&
				$2\sqrt{2}\sqrt{L+14}-5  $&
				$ 8$\\
				Minimum-Redundancy Kl{\o}ve Array (KA$ _R $)&
				yes&
				$(6N^2+36N-15)/23$&
				$ H$&
				$ (3N^2+18N-19)/23$&
				$\sqrt{23/3}\sqrt{L+2}-3  $&
				$ 4N/23+12/23 $\\
				\hline
			\end{tabular}
				\begin{tablenotes}
					\item[a] The tabulated results are representative of the scaling with $ N $ or $ L $ for the parameters maximizing $ H $. The expressions only hold exactly for specific values of $ N $ and $ L $, which may vary between configurations.
					\item[b] The R-MRA has both asymmetric and symmetric solutions. A symmetric solution exists at least for all $ N\leq 48 $.
				\end{tablenotes}
		\end{threeparttable}
	}
\end{table*} 
\begin{table*}[]
	\centering
	\caption{Asymptotic array figures of merit. The KA$_R $ has at most $ 27\% $ more sensors than the R-MRA, which is less than other known arrays with closed-form sensor positions and a contiguous sum co-array.}\label{tab:asymptotic}
	\begin{tabular}{c|c|c|c|c|c|c|c|c}
	\multirow{2}{*}{Array}&
	\multirow{2}{*}{$ R_\infty $}&
	\multirow{2}{*}{$ F_\infty $}&
	\multicolumn{2}{|c|}{$\lim H/H^\txt{R-MRA} $}&
	\multicolumn{2}{|c|}{$ \lim N/N^\txt{R-MRA} $}&
	\multicolumn{2}{|c}{$\lim L/L^\txt{R-MRA} $}\\
	&&&$ N\to\infty $&$ L\to\infty $&$ H\to\infty $&$ L\to\infty $&$ H\to\infty $&$ N\to\infty $\\
		\hline
		MRA&$ 1.09 - 1.73 $&$ 0.5-1 $&$ 1-1.76 $&$ 0.5-1 $&$ 0.75-1$&$ 0.53-1 $&$ 1-2 $&$ 1-3.52 $\\
		NA&$ 2 $&$ 0.5 $&$ 0.60-0.96 $&$ 0.5 $&$1.02-1.30 $&$ 0.72- 0.92$&$ 2 $&$ 1.19-1.92 $\\
		KMA&$ 1.75 $&$ 0.64 $&$0.68-1.10  $&$0.64 $&$ 0.96-1.21$&$ 0.76-0.97 $&$ 1.57$&$  1.07-1.72$\\
		\hdashline
		R-MRA&$ 1.19-1.92 $&$ 1 $&$ 1 $&$ 1 $&\multicolumn{2}{c|}{$ 1$}&$ 1 $&$ 1 $\\
		RRA&$ \infty $&$ 1 $&$ 0 $&$ 1 $&\multicolumn{2}{c|}{$ \infty $}&$ 1 $&$ 0 $\\
		CNA&$ 2 $&$ 1 $&$ 0.60-0.96 $&$ 1 $&\multicolumn{2}{c|}{$ 1.02-1.30 $}&$ 1 $&$ 0.60-0.96 $\\
		KA$ _S$&$ 2$&$ 1 $&$ 0.60-0.96 $&$ 1 $&\multicolumn{2}{c|}{$ 1.02-1.30 $}&$ 1 $&$ 0.60 - 0.96 $\\
		KA$ _R$&$1.92 $&$ 1 $&$ 0.62-1 $&$ 1 $&\multicolumn{2}{c|}{$1- 1.27$}&$ 1 $&$ 0.62 -1 $\\
		\hline
	\end{tabular}
\end{table*} 

\begin{figure*}[]
	\centering
	\subfigure{\includegraphics[width=1\textwidth]{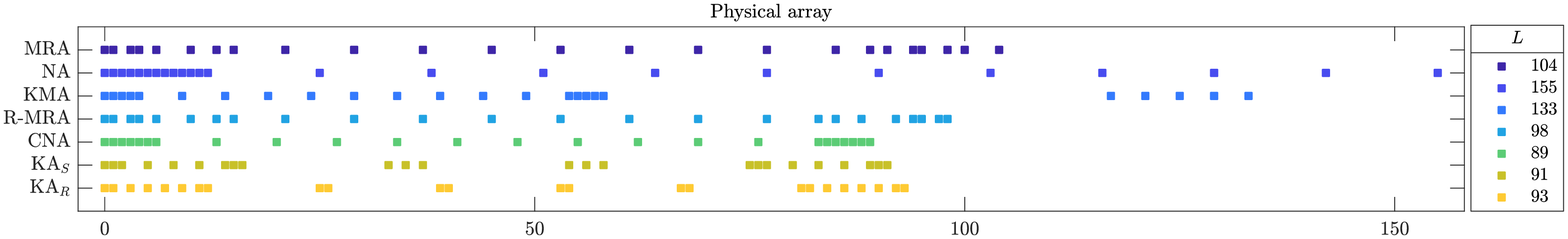}\label{fig:arrays_phys}}
	\hfil
	\subfigure{\includegraphics[width=1\textwidth]{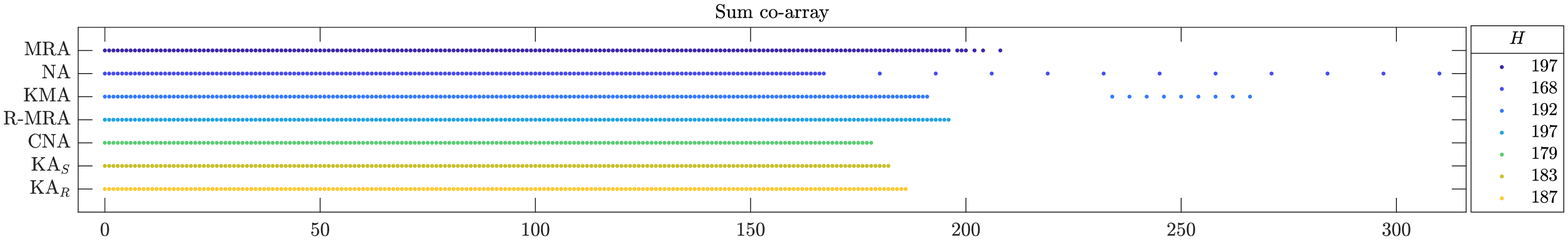}\label{fig:arrays_sca}}
	\hfil
	\subfigure{\includegraphics[width=.95\textwidth]{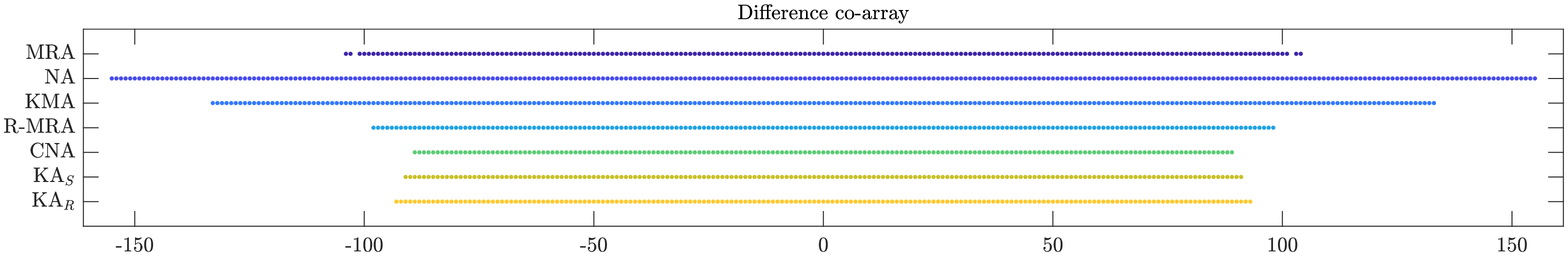}\label{fig:arrays_dca}}
	\caption{Sparse array configurations with $ N=24 $ sensors, and the corresponding sum and difference co-arrays (in units of the smallest inter-sensor spacing). Both co-arrays are contiguous for the R-MRA, CNA, KA$ _S $ and KA$ _R $. The sum and difference co-arrays of any symmetric physical array are merely translated copies of each other.}
	\label{fig:arrays}
\end{figure*}

\subsubsection{Non-contiguous sum co-array}
We first briefly consider the configurations with a non-contiguous sum co-array before studying the arrays with a contiguous sum co-array in more detail. Note from \cref{tab:summary} that even when the sum co-array is non-contiguous, the number of total DoFs $ |\mathcal{D}_\Sigma| $ is only slightly larger than the number of contiguous DoFs $ H $. Specifically, the difference between the two is proportional to the number of physical sensors $ N$, which is much smaller than $ H\propto N^2 $, especially when $ N$ or $H$ grow large. \cref{tab:asymptotic} shows that for the same number of contiguous DoFs $ H\to\infty $, the general MRA has at most $ 25\% $ fewer sensors than the R-MRA by the most conservative estimate. For a fixed aperture $ L\to\infty $, the corresponding number is at most $47\%  $, mainly due to the uncertainty related to $ L $, and thus the asymptotic co-array filling ratio $ F_\infty $, of the general MRA. In particular, any configuration seeking to maximize $ H $,  such as the MRA, must satisfy $ H\geq L$. Moreover, $  H\leq 2L+1  $ holds by definition.

Among the considered configurations with closed-form sensor positions, the KMA achieves the largest number of contiguous DoFs $ H$, and therefore the lowest redundancy\footnote{The construction in \cite{kohonen2017animproved} achieves an approximately $ 1\% $ lower $ R_\infty $.} for a fixed number of sensors $ N $. However, its sum co-array contains holes, since $ F_\infty=7/11< 1$. For equal aperture $ L \to\infty$, the KA$ _R $ has a $57\% $ larger $ H $ than the KMA, but only $31\% $ more physical sensors $ N $. Similarly, the CNA has a two times larger $ H $ than the NA and $ 41\% $ larger $ N $. Conversely, for equal $ N \to\infty$, the CNA has a $ 50\% $ smaller physical aperture $ L $ than the NA, while still achieving the same $ H $. The KA$ _R $ has a $ 42\% $ smaller $ L $ than the KMA, but only a $ 9\% $ smaller $ H $.

\subsubsection{Contiguous sum co-array}
Next, we turn our attention to the configurations with a contiguous sum co-array. \cref{fig:N_vs_L} illustrates the array aperture $ L $ as a function of the number of sensors $ N $. The aperture scales quadratically with $ N $ for all configurations except the RRA with $L\propto N $. For reference, recall that the ULA satisfies $ L = N-1 $. The KA$_R  $ achieves a slightly larger aperture than the rest of the considered parametric arrays. For the same number of sensors (approaching infinity), the KA$ _R $ has a $ 4\% $ larger aperture than the KA$ _S $ and CNA. Conversely, for a fixed aperture, the KA$ _R $ has $ 2\% $ fewer sensors. The differences between the configurations is clear in \cref{fig:N_vs_R}, which shows the redundancy $ R $ as a function of $ N $. By definition, the R-MRA is the least redundant array with a contiguous sum co-array. However, it is also computationally prohibitively expensive to find for large $ N $, and therefore only known for $ N \leq 48 $ \cite{kohonen2014meet,kohonen2015early}. For $ N \geq 49 $, one currently has to resort to alternative configurations that are cheaper to generate, such as the KA$ _R $. The KA$_R$ achieves the lowest redundancy when $ N\geq 72$. When $ 49 \leq N\leq 71$, the RRA is the least redundant configuration. However, the redundancy of the RRA goes to to infinity with increasing $ N $. The KA$ _R $ has between $0$ and $ 27\% $ more sensors than then R-MRA, which is less than any other currently known parametric sparse array configuration with a contiguous sum co-array.

\begin{figure}[]
	\centering
	\includegraphics[width=1\linewidth]{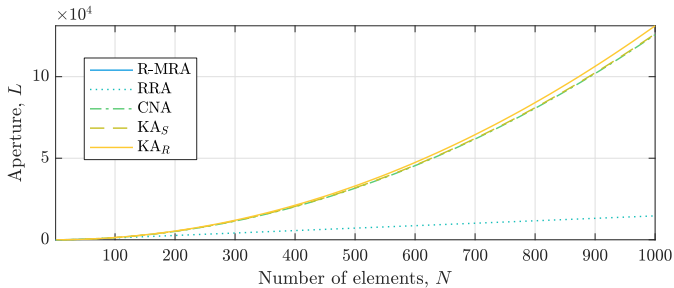}
	\caption{The aperture of the sparse arrays grows quadratically with the number of sensors $ N $. For a given $ N $, the KA$ _R $ has the largest aperture of all currently known parametric arrays with a contiguous sum co-array.}
	\label{fig:N_vs_L}
\end{figure}

\begin{figure}[]
	\includegraphics[width=1\linewidth]{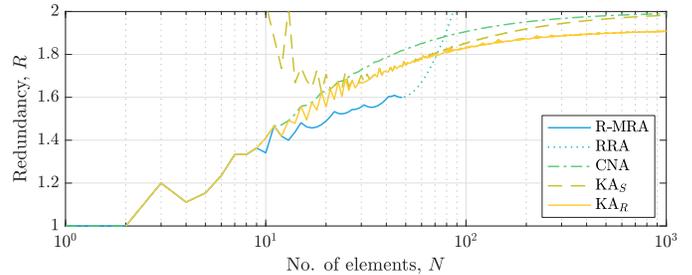}
	\caption{The KA$ _R $ achieves the lowest redundancy $ R $ for $ N\geq 72 $ sensors. When $ 49\leq N\leq 71 $, the RRA is less redundant. The R-MRA is the least redundant configuration with a contiguous sum co-array for any $ N $, but it is computationally
expensive to find and unknown for $ N\geq  49 $.}\label{fig:N_vs_R}
\end{figure}

\cref{fig:N_vs_S1} shows the number of unit spacings $ S(1) $ as a function of $ N $. In general, $ S(1) $ increases linearly with $ N $, and the KA$ _R $ has the smallest rate of growth. The two exceptions are the KA$ _S $ and RRA, which have a constant $ S(1) $. However, unlike the RRA, the redundancy of the KA$ _S $ is bounded (cf.~\cref{fig:N_vs_R}). As discussed in \cref{sec:fom_S}, the number of unit spacings $ S(1) $ may be used as a simplistic indicator of the robustness of the array to mutual coupling. Assessing the effects of coupling ultimately requires actual measurements of the array response, or simulations using an electromagnetic model for the antenna and mounting platform \cite{allen1966mutual,rubio2015mutual,craeye2011areview,friedlander2020theextended}. A detailed study of mutual coupling is therefore beyond the scope of this paper.
\begin{figure}[]
	\centerline{\includegraphics[width=1\linewidth]{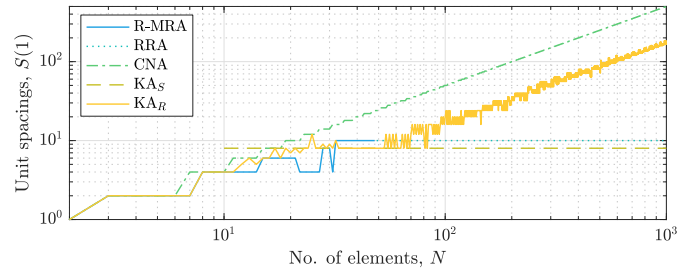}}
	\caption{The number of unit spacings $ S(1) $ of the KA$ _R $ grows linearly with the number of sensors $ N $. Both the RRA and KA$ _S $ have a constant $ S(1) $. However, the RRA does not have a bounded asymptotic redundancy.}\label{fig:N_vs_S1}
\end{figure}

Obviously, many other important figures of merit are omitted here for brevity of presentation. For example, \emph{fragility} and the achievable \emph{beampattern} are natural criteria for array design or performance evaluation. Fragility quantifies the sensitivity of the co-array to physical sensor failures \cite{liu2019robustness}. The array configurations studied in this paper only have \emph{essential} sensors, and therefore high fragility, since the difference (and sum) co-array ceases to be contiguous if a sensor is removed. This is the cost of low redundancy. The beampattern is of interest in applications employing linear processing\footnote{For exceptions where the beampattern is also relevant when employing non-linear processing, see, e.g., \cite{pal2010nested,cohen2018sparseconvolutional}.}. For example, in adaptive beamforming, the one-way (transmit \emph{or} receive) beampattern is critical, whereas in active imaging, the two-way (combined transmit \emph{and} receive) beampattern is more relevant. Although the one-way beampattern of a sparse array generally exhibits high sidelobes, a wide range of two-way beampatterns may be achieved using one \cite{cohen2020sparse} or several \cite{hoctor1990theunifying,kozick1991linearimaging} transmissions and receptions. The arrays discussed in this paper can achieve the same effective beampattern as the ULA of equivalent aperture by employing multiple transmissions and receptions.

\subsection{Active sensing performance}
The estimation of the angles and scattering coefficients from \eqref{eq:z} can be formulated as an on-grid \emph{sparse support recovery} problem, similarly to the passive sensing case described in \cite{pal2015pushing}. In particular, let $ \{\tilde{\varphi}_i\}_{i=1}^V $ denote a set of $V\gg K  $ discretized angles. The task then becomes to solve
\begin{opteq}
	\underset{\bm{\bar{\gamma}}\in\mathbb{C}^{V}}{\txt{minimize}}&\ \|\bm{x}-(\bm{\tilde{A}}\odot \bm{\tilde{A}})\bm{\tilde{\gamma}}\|_2^2\
	\txt{subject to}\ \|\bm{\tilde{\gamma}}\|_0=K, \label{p:cs}
\end{opteq}
where $ \bm{\tilde{A}} \in\mathbb{C}^{N\times V}$ is the known steering matrix sampled at the $ V $ angles,  and $ \bm{\bar{\gamma}} \in\mathbb{C}^V$ the unknown sparse scattering coefficient vector. The sparsity of $ \bm{\bar{\gamma}} $ is enforced by the $ \ell_0 $ pseudonorm $ \|\bm{\tilde{\gamma}}\|_0 \triangleq \sum_{i=1}^V\mathbbm{1}(\tilde{\gamma}_i\neq 0)$, which enumerates the number of non-zero entries. Although \eqref{p:cs} is a non-convex optimization problem due to the $ \ell_0 $ pseudonorm, it can be approximately solved using \emph{Orthogonal Matching Pursuit} (OMP). OMP is a greedy algorithm that iteratively finds the $k = 1,2\ldots, K $ columns in the dictionary matrix $ \bm{\tilde{A}}\odot \bm{\tilde{A}} $ whose linear combination best describes the data vector $ \bm{x} $ in an $ \ell_2 $ sense. For details on the OMP algorithm, see \cite{tropp2004greed}, \cite[p.~65]{foucart2013amathematical}.

We now compare the active sensing performance of the KMA and KA$ _R $ using OMP. As demonstrated in \cref{sec:comparison}, the KMA and KA$ _R $ achieve the lowest redundancy of the considered configurations with closed-form sensor positions. We consider $ K= 65$ scatterers with angles $ \varphi_k $ uniformly spaced between $ -60^\circ $ and $ 60^\circ $. The scattering coefficients are spherically distributed complex-valued random variables, i.e., $ \gamma_k=z_k/|z_k| $, where $ z\sim\mathcal{CN}(0,1) $. We assume that the arrays consist of ideal omnidirectional sensors with unit inter-sensor spacing $ \delta = 1/2 $, such that the $ (n,i) $th entry of the sampled steering matrix is $ \tilde{A}_{n,i} = \exp(j\pi d_n \sin\tilde{\varphi}_i) $. Angles $ \{\tilde{\varphi}_i\}_{i=1}^V $ form a uniform grid of $ V\!=\!10^4 $ points between $ -75^\circ $ and $ 75^\circ $.

\cref{fig:arrays2} shows that for a fixed aperture $ L=70 $, the KA$ _R $ has $ 141/101\approx 40\% $ more contiguous DoFs than the KMA, at the expense of $ 21/17\approx 24\% $ more sensors\footnote{For reference, the R-MRA with aperture $ L=70 $ has $ N=20 $ sensors.}. The KA$ _R $ can therefore resolve more scatterers than the KMA of equivalent aperture at little extra cost. This is illustrated in \cref{fig:omp_wo_noise}, which shows the noiseless spatial spectrum estimate produced by the OMP algorithm. Both arrays resolve more scatterers than sensors, although the KMA misses some scatterers and displays spurious peaks instead. The KA$_R $ approximately resolves all scatterers due to its larger co-array. \cref{fig:omp_wo_noise} shows the noisy OMP spectra when the SNR of the scene, defined as $ K/\sigma^2 $, is $ 5 $ dB. The spectrum degrades more severely in case of the KMA, partly because it has fewer physical sensors than the KA$ _R $. These conclusions are validated by the empirical \emph{root-mean square error} (RMSE) of the angle estimates, where $ 1000 $ random realizations of the noise and scattering coefficients were used. In the noiseless case, the RMSE of the KMA is $ 8.6^\circ $, whereas the KA$_R  $ achieves $ 4.2^\circ $. In the noisy case the corresponding figures are $ 16.2^\circ $ (KMA) and $ 7.9^\circ $ (KA$ _R $).
\begin{figure}[]
	\centering
	\subfigure{\includegraphics[width=1\linewidth]{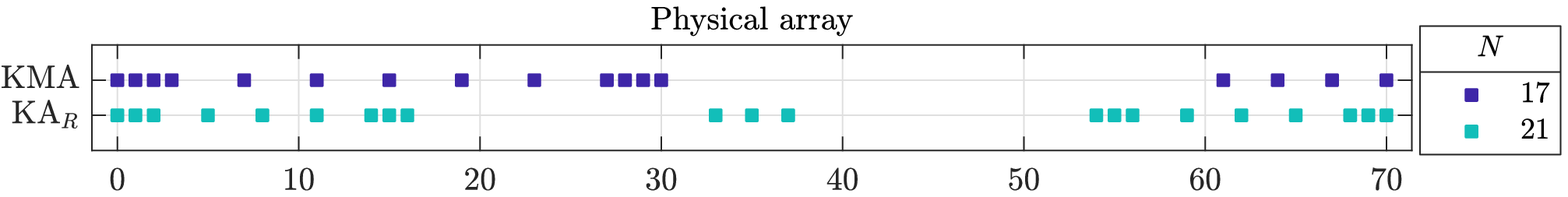}\label{fig:arrays_phys}}
	\hfil
	\subfigure{\includegraphics[width=1\linewidth]{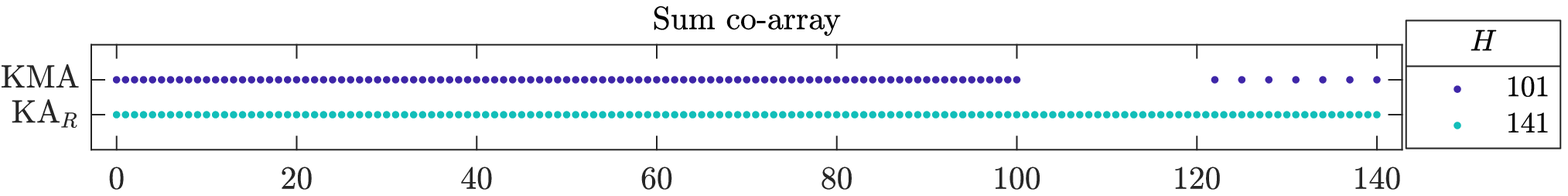}\label{fig:arrays_sca}}
	\caption{Equi-aperture array configurations. The KA$ _R $ has $ 4 $ more physical sensors than the KMA, but $ 40 $ more contiguous elements in its sum co-array.}
	\label{fig:arrays2}
\end{figure}

\begin{figure}[]
	\centering
	\subfigure{\includegraphics[width=1\linewidth]{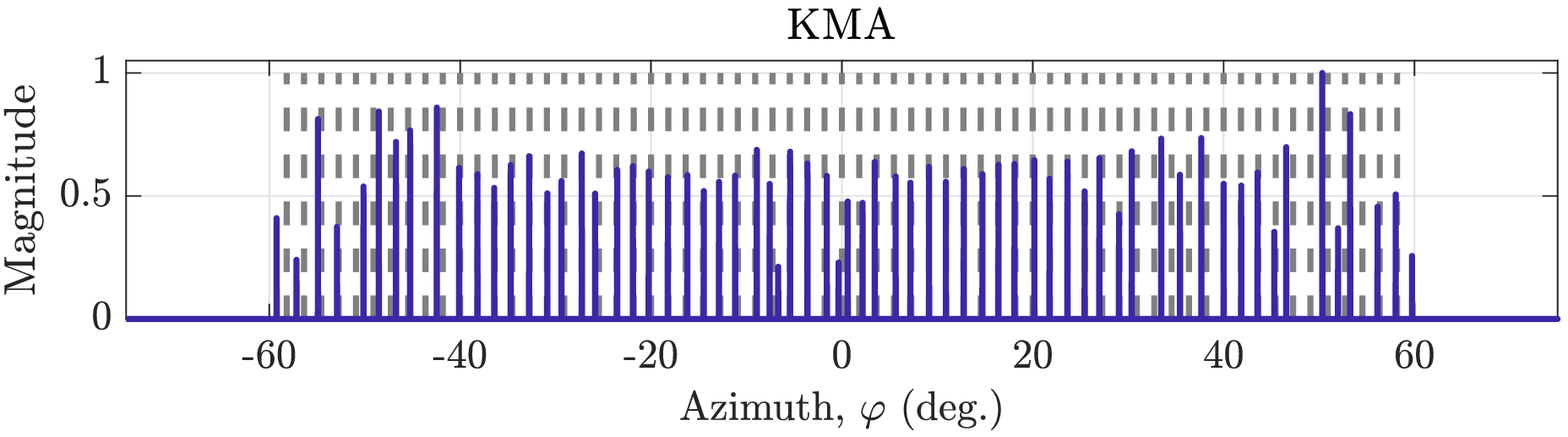}\label{fig:1}}
	\hfil
	\subfigure{\includegraphics[width=1\linewidth]{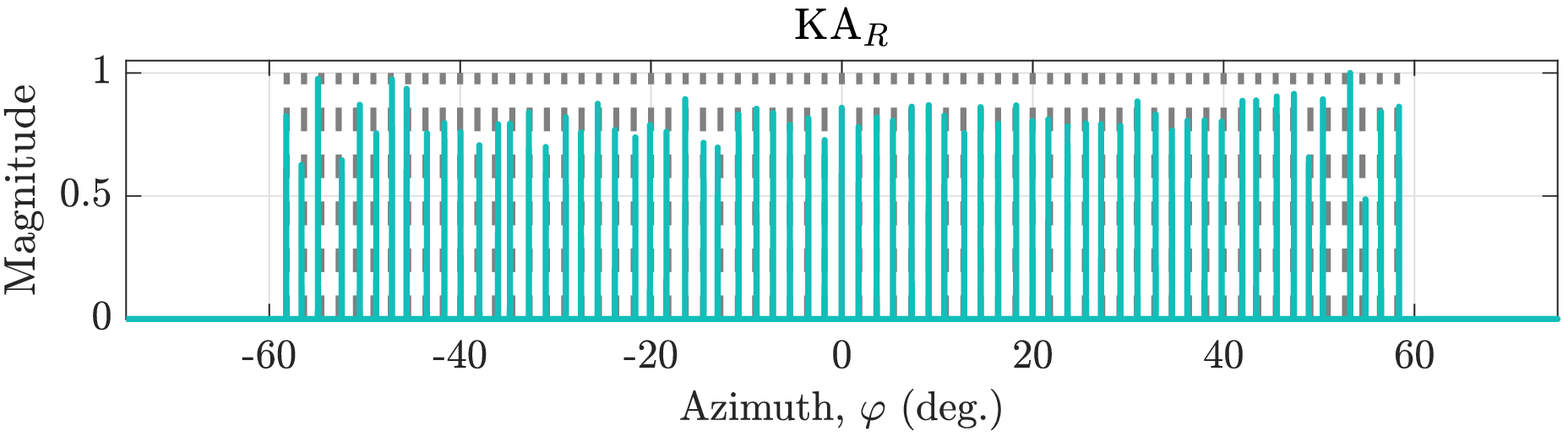}\label{fig:2}}
	\caption{Noiseless spatial spectrum estimate. The KA$ _R $ resolves more scatterers than the KMA, due to its larger sum co-array. Both arrays find more scatterers than sensors. The dashed lines indicate the true scatterer directions.}
	\label{fig:omp_wo_noise}
\end{figure}

\begin{figure}[]
	\centering
	\subfigure{\includegraphics[width=1\linewidth]{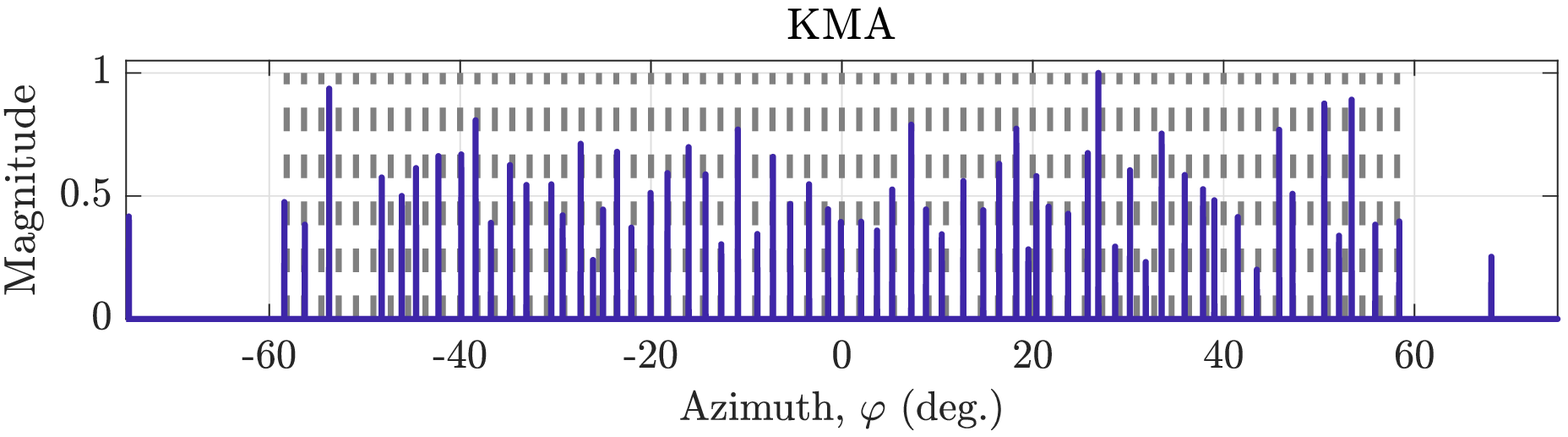}\label{fig:1}}
	\hfil
	\subfigure{\includegraphics[width=1\linewidth]{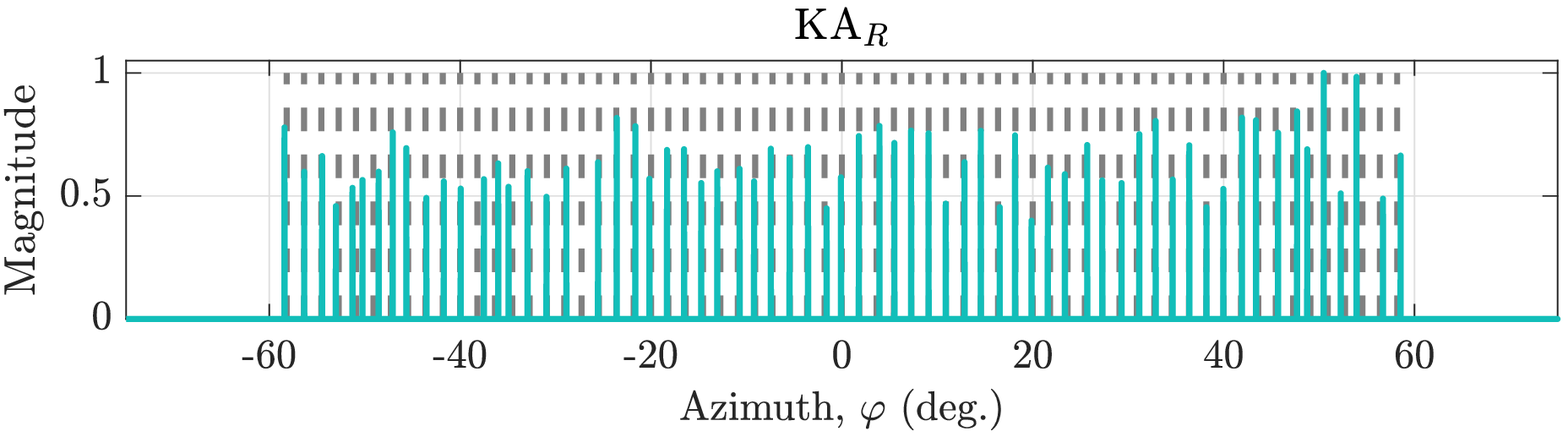}\label{fig:2}}
	\caption{Noisy spatial spectrum estimate. The KMA displays multiple false peaks as it has fewer physical (and sum co-array) elements than the KA$ _R $.}
	\label{fig:omp_w_noise}
\end{figure}

\section{Conclusions and future work} \label{sec:conclusions}
This paper proposed a general symmetric sparse linear array design suitable for both active and passive sensing. We established a necessary and sufficient condition for the sum and difference co-array to be contiguous, and identified sufficient conditions that substantially simplify the array design. We studied two special cases in detail, the CNA and KA, both of which achieve a low redundancy and can be generated for any number of sensors $ N $. The KA achieves the lowest asymptotic redundancy among the considered array configurations. This also yields an upper bound on the redundancy of the R-MRA, whose exact value remains an open question. The upper bound may perhaps be tightened by novel sparse array designs suggested by the proposed array design methodology.

In future work, it would be of interest to characterize the redundancy of other symmetric arrays, including recursive/fractal arrays \cite{yang2018aunified,cohen2019sparsefractal,cohen2020sparse} that have a contiguous sum co-array. Another direction is investigating the advantages of symmetric arrays over co-array equivalent asymmetric arrays in more detail. This could further increase the relevance of the symmetric sparse array configurations studied in this paper.

\section*{Acknowledgment}
The authors would like to thank Dr. Jukka Kohonen for bringing the Kl{\o}ve basis to their attention and for the feedback on this manuscript, as well as for the many stimulating discussions regarding additive bases.

\appendices

\section{Contiguous difference co-array of the KMA} \label{a:KMA_diff_coarray}
We now show that the Kl{\o}ve-Mossige array $\mathcal{G}$ in \cref{def:S-KMA} has a contiguous difference co-array. By symmetry of the difference co-array, it suffices to show that $ \mathcal{G}-\mathcal{G}\supseteq\mathcal{C}\supseteq\{0:\max\mathcal{G}\} $. We may easily verify that $ \mathcal{G}-\mathcal{G}\supseteq\mathcal{C} $ holds for
\begin{align*}
\mathcal{C}&=(\mathcal{D}_\txt{CNA}-\mathcal{D}_\txt{CNA})\cup (\mathcal{D}_3-\mathcal{D}_\txt{CNA}+2\max\mathcal{D}_\txt{CNA}+1)\\
&\supseteq\{0:\max\mathcal{D}_\txt{CNA}\}\cup (\mathcal{D}_3+\mathcal{D}_\txt{CNA}+\max\mathcal{D}_\txt{CNA}+1),
\end{align*}
where the second line follows from the fact that the CNA is symmetric and has a contiguous difference co-array. Consequently, $ \mathcal{C}\supseteq\{0:\max\mathcal{G}\} $ holds if and only if
\begin{align*}
\mathcal{D}_3+\mathcal{D}_\txt{CNA}=\{0:\max\mathcal{G}-\max\mathcal{D}_\txt{CNA}-1\}.
\end{align*}
Due to the periodicity of $ \mathcal{D}_3 $, this condition simplifies to
\begin{align*}
\mathcal{D}_4+\mathcal{D}_\txt{CNA}= \{0:\max\mathcal{D}_\txt{CNA}+N_1^2\},
\end{align*}
where $ \mathcal{D}_4=\{0:N_1:N_1^2\} $, and by \cref{def:CNA} we have
\begin{align*}
\mathcal{D}_4+\mathcal{D}_\txt{CNA}=\{\mathcal{D}_4+\mathcal{D}_1\}&\cup\{\mathcal{D}_4+\mathcal{D}_2+N_1\}\\ &\cup\{\mathcal{D}_4+\mathcal{D}_1+N_2(N_1+1)\}.
\end{align*}
As $ \mathcal{D}_1+\mathcal{D}_4=\{0:N_1(N_1+1)-1\} $, it suffices to show that
\begin{align*}
\mathcal{D}_4+\mathcal{D}_2 &\supseteq \{N_1^2:(N_2-1)(N_1+1)\}.
\end{align*}
By \cref{def:CNA,def:S-KMA}, we have
\begin{align*}
\mathcal{D}_4+\mathcal{D}_2 &=\{kN_1+l(N_1+1)\ |\ k\in\{0: N_1\}; l\in\{0: N_2-1\}\}\\
&=\{i(N_1+1)-k\ |\ k\in\{0: N_1\};i-k\in\{0: N_2-1\}\}\\
&\supseteq\{i(N_1+1)-k\ |\ k\in\{0: N_1\};i\in\{N_1: N_2-1\}\}\\
&\supseteq \{N_1^2:(N_2-1)(N_1+1)\},
\end{align*}
which implies that the difference co-array of $\mathcal{G}$ is contiguous.

\section{First hole in the sum co-array of the KMA}\label{a:S-KMA_lambda}
Let $\mathcal{G}$ denote the KMA in \cref{def:S-KMA}. Furthermore, let $ H\in\mathbb{N} $, as defined in \eqref{eq:H}, be the first hole in $\mathcal{G}+\mathcal{G}$. In the following, we show that
\begin{align*}
	H = \begin{cases}
	2\max\mathcal{G}+1,&\txt{if } N_1 +N_2=1\\
	h+1,&\txt{if } N_1 \geq 1 \txt{ and } N_2 = 1\\
	h,&\txt{otherwise},
	\end{cases}
\end{align*}
where the non-negative integer $ h $ is
\begin{align}
h&=\max\mathcal{G}+\max\mathcal{D}_\txt{CNA}+1\nonumber\\
&=N_3(\max\mathcal{D}_\txt{CNA}+1+N_1^2) +2\max\mathcal{D}_\txt{CNA}+1. \label{eq:h}
\end{align}
The first case, which we only briefly mention here, follows trivially from the fact that $\mathcal{G}$ degenerates to the ULA when either $N_1=0 $ and $ N_2 = 1 $, or $ N_1=1 $ and $ N_2=0 $. We prove the latter two cases by contradiction, i.e., by showing that $h +1\in \mathcal{G}+\mathcal{G} $, respectively $h \in \mathcal{G}+\mathcal{G} $, leads to an impossibility. Verifying that indeed $\mathcal{G}+\mathcal{G}\supseteq\{0:h\} $, respectively $\mathcal{G}+\mathcal{G}\supseteq\{0:h-1\} $, is left as an exercise for the interested reader.

We start by explicitly writing the sum co-array of $\mathcal{G}$ as
\begin{align*}
\mathcal{G}+\mathcal{G} = (\mathcal{D}_\txt{CNA}+\mathcal{D}_\txt{CNA})&\cup(\mathcal{D}_\txt{CNA}+\mathcal{D}_3+2\max\mathcal{D}_\txt{CNA}+1)\\
&\cup (\mathcal{D}_3+\mathcal{D}_3+4\max\mathcal{D}_\txt{CNA}+2).
\end{align*} 
Note that the CNA has a contiguous sum co-array, that is,
\begin{align*}
\mathcal{D}_\txt{CNA}+\mathcal{D}_\txt{CNA}=\{0:2\max\mathcal{D}_\txt{CNA}\}.
\end{align*}
Furthermore, it was shown in Appendix~\ref{a:KMA_diff_coarray} that
\begin{align*}
\mathcal{D}_\txt{CNA}+\mathcal{D}_3+2\max\mathcal{D}_\txt{CNA}+1 = \{2\max\mathcal{D}_\txt{CNA}+1:h-1\}.
\end{align*}
Consequently, $h\in \mathcal{G}+\mathcal{G} $ holds if and only if
\begin{align*}
h\in \mathcal{D}_3+\mathcal{D}_3+4\max\mathcal{D}_\txt{CNA}+2.
\end{align*}
By \cref{def:S-KMA}, there must therefore exist non-negative integers $ k\in\{0:2N_1\}$ and $ l\in\{0:2(N_3-1)\} $ such that 
\begin{align}
h\!=\!kN_1+l(\max\mathcal{D}_\txt{CNA}+1+N_1^2)+4\max\mathcal{D}_\txt{CNA}+2.\label{eq:a2_c1}
\end{align}
Substituting \eqref{eq:h} into \eqref{eq:a2_c1} and rearranging the terms yields
\begin{align*}
(N_3-l)(\max\mathcal{D}_\txt{CNA}+1+N_1^2)=2\max\mathcal{D}_\txt{CNA}+1+kN_1.
\end{align*}
Since $ k\in\{0:2N_1\}$, the following inequality must hold:
\begin{align*}
\frac{2\max\mathcal{D}_\txt{CNA}+1}{N_1^2+\max\mathcal{D}_\txt{CNA}+1} \leq N_3-l \leq \frac{2N_1^2+2\max\mathcal{D}_\txt{CNA}+1}{N_1^2+\max\mathcal{D}_\txt{CNA}+1}.
\end{align*}
This reduces to $0 < N_3-l< 2 $, or more conveniently, $ N_3-l =1$, since $ N_3-l $ is an integer. Consequently, we have
\begin{align}
N_1(N_1-k)=\max\mathcal{D}_\txt{CNA},\label{eq:a2_c2}
\end{align}
where $ \max\mathcal{D}_\txt{CNA} \geq 0$ leads to $ N_1-k \in\{0:N_1\} $. Substituting $\max\mathcal{D}_\txt{CNA}=L$ in \eqref{eq:L_CNA} into \eqref{eq:a2_c2} yields
	\begin{align}
	N_1-k = N_2+1+\frac{N_2-1}{N_1}. \label{eq:a2_c3}
	\end{align}
	Combined with $ N_1-k\leq N_1 $, this implies that
	\begin{align*}
	N_1\geq \frac{N_2+1+\sqrt{(N_2+1)^2+4(N_2-1)}}{2}\geq N_2+1,
	\end{align*}
	since $ N_1,N_2\geq 1 $. We identify the following two cases:
	\begin{enumerate}[label=\roman*)]
	 \item If $ N_2=1$, then \eqref{eq:a2_c3} yields that $ N_1-k=2$, implying that $H> h $. However, it is straightforward to verify that $H= h+1 $ from the fact that when $ h $ is replaced by $ h+1 $ in \eqref{eq:a2_c1}, no integer-valued $ N_1\geq 2 $ satisfies the equation.
 	\item If $ N_2\geq 2 $, then $ N_1 \leq  N_2-1$ follows from \eqref{eq:a2_c3}, since $ (N_2-1)/N_1 $ must be an integer. This leads to a contradiction, since both $ N_1\geq N_2+1$ and $ N_1\leq N_2-1$ cannot hold simultaneously. Consequently, $H= h $ holds.
	\end{enumerate}
Finally, $ H=h $ also holds when $ N_1=0 $ and $ N_2 \geq 2$, or $ N_1\geq 2 $ and $ N_2=0 $, since $\mathcal{G}$ degenerates into the NA in this case. This covers all of the possible values of $ H $.

\bibliographystyle{IEEEtran}
\bibliography{IEEEabrv,bibliography}

\end{document}